\documentclass[a4paper,11pt]{article}

\usepackage{amsthm,amsmath,amssymb,amsfonts,graphics,bbm,multirow,pifont}
\usepackage[colorlinks]{hyperref}
\theoremstyle{plain}
\newtheorem{theorem}{Theorem}
\newtheorem{Pro}[theorem]{Proposition}
\newtheorem{prop}[theorem]{Proposition}
\newtheorem{The}[theorem]{Theorem}
\newtheorem{Lem}[theorem]{Lemma}
\newtheorem{Cor}[theorem]{Corollary}

\theoremstyle{definition}
\newtheorem{example}[theorem]{Example}

\newtheorem{definition}[theorem]{Definition}

\def\Z{\mathbb{Z}}	
\def\C{\mathbb{C}}

\def\N{{\mathbb N}}

\def\cT{\mathcal{T}}
\def\cS{{\mathcal S}}
\def\bbbn{{\mathbb N}}
\def\bbbz{{\mathbb Z}}

\def\r{{\rm r}}
\def\l{{\rm l}}
\def\c{{\rm c}}

\renewcommand{\leq}{\leqslant} 		
\renewcommand{\geq}{\geqslant}

\def\bbbc{{\mathbb C}}

\newcommand\cD{{\mathcal D}}

\def\a{\alpha}
\def\b{\beta}
\def\p{\partial}
\def\om{\omega}

\def\d{\partial}
\def\bd{\bar{\partial}}
\def\fieldk{\bbbc}
\def\ff{\mathfrak{f}}

\def\fA{\mathfrak{A}}
\def\fF{\mathfrak{F}}
\def\fB{\mathfrak{B}}
\def\cA{\mathcal{A}}
\def\cB{\mathcal{B}}
\def\cT{\mathcal{T}}
\def\cS{{\mathcal S}}
\def\cZ{\mathcal{Z}}

\def\fK{\mathfrak{K}}
\def\fI{\mathfrak{I}}
\def\fJ{\mathfrak{J}}

\def\cI{\mathcal{I}}
\def\cJ{\mathcal{J}}

\begin{document}

\bibliographystyle{unsrt}
\title{Algebraic quantisation approach to integrable differential-difference equations}
\author{Sylvain Carpentier $^\ddagger$, Alexander V. Mikhailov$^{\star}$ and
Jing Ping 
Wang$ ^\dagger $
\\
$\ddagger$ Department of Mathematical Sciences, Seoul National University, South Korea
\\
$\star$ School of Mathematics, University of Leeds, UK\\
$\dagger$ School of Mathematics and Statistics, Ningbo University, \\Ningbo 315211,
People’s Republic of China
}
\date{}  
\maketitle
\begin{abstract}
 We develop an algebraic quantisation approach, based on quantisation ideals, and apply it to integrable non-Abelian differential--difference equations. We show that the Toda hierarchy admits a bi-quantum structure whose classical (commutative) limit recovers a well-known Poisson pencil. In addition, we discover a non-standard quantisation that has no commutative counterpart. In both cases we present the quantum systems in the Heisenberg form.
 The generality of the method is illustrated through a wide range of integrable lattices, including the modified Volterra, Bogoyavlensky, Ablowitz-Ladik, relativistic Toda, Merola-Ragnisco-Tu, Adler-Yamilov, Chen-Lee-Liu, Belov-Chaltikian, and Blaszak-Marciniak systems. For each of them, we construct explicit quantisation ideals and present the first few commuting quantum Hamiltonians.
\end{abstract}

%%%%%%%%%%%%%%%%%%%%%%%%%%%%%%%%%%%%%%%%%%

{\small
\tableofcontents
}

\section{Introduction}

The original idea of quantisation, introduced by Heisenberg, was to replace commutative variables in classical phase space with non-commutative operators, and their multiplication depends on the quantisation parameter, the Planck constant $\hbar$ \cite{Heisenberg25}. Dirac observed that in the classical limit $\hbar \to 0$, commutators of operators reduce to Poisson brackets of the corresponding classical observables. He introduced the concept of a {\em quantum algebra} and raised foundational questions concerning operator ordering and the consistency of quantisation with dynamics at finite values of $\hbar$  \cite{Dirac25}. 

These ideas provided motivation and a starting point for the theory of deformation quantisation \cite{BFFLS}. Kontsevich asserted that any smooth Poisson manifold admits a non-commutative deformation of its algebra of functions. The deformed product is associative and is given by a formal power series in $\hbar$ with explicitly computable coefficients \cite{Kontsevich2003}. Despite this breakthrough, issues such as the convergence of these formal series, and ambiguities in operator ordering remain unresolved.
In a recent series of lectures, Witten remarked that ``there actually is no completely natural operation of quantising a classical phase space - none that is known, and I believe, none that exists. Quantisation always requires some additional structure. As Ludwig Faddeev used to say, quantisation is an art, not a science'' \cite{Witten2021}.

An alternative approach to the construction of quantum dynamical systems was recently proposed by Mikhailov \cite{AvM20}. In this algebraic quantisation framework, one starts with a dynamical or differential-difference system defined in a free associative algebra $\fA$, generated by a (possibly infinite) set of generators. Such a system determines a derivation $\partial_t : \fA \to \fA$.
A two-sided ideal $\cI \subset \fA$ is called a \emph{quantisation ideal} for $(\fA,\partial_t)$ if it satisfies the following two conditions:
\begin{itemize}
\item[(i)] $\cI$ is stable under $\partial_t$, and
\item[(ii)] the quotient algebra $\fA/\cI$ admits a basis consisting of normally ordered monomials.
\end{itemize}
The quantum algebra is then defined as the quotient $\fA/\cI$. Condition (i) ensures that $\partial_t$ descends to a derivation of $\fA/\cI$, while condition (ii) guarantees that the resulting commutation relations are encoded by the basis of the quantisation ideal $\cI$, thereby defining the quantum dynamical system corresponding to the ideal $\cI$.

To apply algebraic quantisation to a classical dynamical system with commutative variables, the system must first be lifted to a free associative algebra. This lifting is generally non-unique due to ordering ambiguities in non-commutative variables. However, requiring the lifted system to preserve key structural
features, such as symmetries, significantly reduces this ambiguity. Integrable systems, which possess a hierarchy of commuting symmetries, are particularly restricted on the 
admissible orderings if requiring to preserve the integrability in the non-commutative setting.  
Classification results for some types of non-Abelian integrable differential-difference equations were recently obtained in \cite{NW25}.

This algebraic quantisation approach has been successfully applied to several integrable systems, including:
\begin{itemize}
\item the hierarchy of stationary Korteweg–de Vries equations and Novikov’s equations \cite{BM2021},
\item the full hierarchy of symmetries of the Volterra lattice \cite{CMW, CMW2},
\item the Volterra–Zhukovsky top (the Euler top in a constant external field) \cite{MikSkryp}, and
\item mutation maps arising as solutions of the Zamolodchikov tetrahedral equation \cite{ChMikTal}.
\end{itemize}
Beyond standard quantisations, which admit a classical commutative limit, this framework also uncovers unexpected non-standard quantisation ideals, in which the quantum algebra remains non-commutative for all specialisations of the quantisation parameters. The quantisation of the Volterra hierarchy (see Example \ref{EG1} in Section \ref{sec22}) provides an example where both types of ideals arise.

In this paper, we develop an algebraic quantisation approach and demonstrate its effectiveness. In Section \ref{sec2}, we introduce the necessary notation and basic definitions, including the concept of local functionals.
We provide a comprehensive treatment of the non-Abelian Toda hierarchy, establishing two key results: bi-quantisation and non-standard quantisation in  Section \ref{sec3}.

The classical Toda lattice was introduced by Toda in 1967. Its non-Abelian  version appeared in \cite{bmrl80, k81, mik81}. In evolutionary form, the integrable non-Abelian Toda lattice is given by
\begin{equation}\label{Todaintro}
	\begin{cases}
	 \partial_{t_1}{a_n} = b_{n+1} a_n - a_n b_n\\
	 \partial_{t_1}b_n = a_n - a_{n-1}
	\end{cases}
\end{equation}
where $a_n$ and $b_n$ are matrix-valued variables.

Quantisation of the Toda lattice was studied in \cite{a} using the $r$-matrix approach, which requires an ultralocal change of variables. In this framework, quantum Hamiltonians are obtained from products of ultralocal Lax operators, while commutation relations arise as deformations of the canonical Poisson brackets.

Our algebraic quantisation approach differs greatly from this traditional r-matrix approach.
It does not rely on the Poisson structure of the classical system, nor does it involve any change of variables.
In Section \ref{sec33}, we show that the Toda lattice \eqref{Todaintro} along with the entire hierarchy of its commuting symmetries:
\begin{equation}\label{Todahieintro}
	\partial_{t_k}{a_n} = K^{(k)}, \qquad \partial_{t_k}b_n = P^{(k)}, \qquad k \in \mathbb{N},
\end{equation}
where $k=1$ corresponds to the Toda lattice \eqref{Todaintro},
admits a quantisation ideal generated by the polynomials:
\begin{equation*}
 \begin{array}{lll}
 b_nb_{n+1}-b_{n+1}b_n-(\omega-1)\, a_n,\qquad
&b_nb_m-b_mb_n,\quad &
|n-m|\ne  1, \\
 a_na_{n+1}-\omega a_{n+1} a_n,\qquad &a_na_m-a_ma_n,\quad &
|n-m|\ne  1, \\
 b_na_n-\omega a_nb_n-\eta a_n,&a_nb_m-b_ma_n,&m-n\ne0,1,\\
    a_nb_{n+1}-\omega b_{n+1}a_n-\eta a_n,& n, m\in\bbbz .& 
 \end{array}
\end{equation*}
These commutation relations depend on two complex quantisation parameter $\omega$ and $\eta$ with $\omega \neq 0$. 
After setting $\omega=1 + \mu \a$, $\eta= \mu \b$ and taking the classical limit $\mu \to 0$,
the algebra of observables becomes commutative and is equipped with a pencil of compatible Poisson brackets. 
It is natural to refer to such systems and algebras as bi-quantum.

In Section \ref{sec34}, we prove that the non-Abelian Toda sub-hierarchy 
\[\partial_{t_{2k}}a_{n}=K^{(2k)}, \qquad \partial_{t_{2k}}b_{n}=P^{(2k)},\qquad k\in\bbbn\]  
admits an alternative, non-standard quantisation with the following commutation relations
\begin{equation*}
 \begin{array}{rclll}
 \phantom{e^{-i\hbar}} b_nb_{n+1}+b_{n+1}b_n&=&(1- \omega)\, a_n,\qquad
&b_nb_m+b_mb_n=0,\quad &
|n-m|\ne  1, \\
 a_na_{n+1} + \omega a_{n+1} a_n&=&0,\qquad &a_na_m-a_ma_n=0,\quad &
|n-m|\ne  1, \\
  b_na_n + \omega a_nb_n&=& 0,&a_nb_m-b_ma_n=0,&m-n\ne0,1,\\
    a_nb_{n+1}+ \omega b_{n+1}a_n&=&0,& n, m\in\bbbz ,&
 \end{array}
\end{equation*}
where $\omega\neq 0$ is a quantisation parameter. For any value of the parameter $\omega$, the corresponding quantum algebra is non-commutative. It may be viewed as a deformation of the non-commutative algebra $\hat\fA $ obtained by setting $\omega = 1$. The equations of this sub-hierarchy are well defined in $\hat\fA $ and can be interpreted as non-abelian Hamiltonian systems equipped with Poisson algebra and Poisson module structures \cite{MV} (see also \cite{lrs25}, which provides an alternative interpretation of  a similar limit). Quantum systems that do not admit a classical commutative limit are common in physics. For example, in the limit $\hbar \to 0$, quantum systems with bosonic and fermionic degrees of freedom become defined on Grassmann algebras, equipped with a $\mathbb{Z}_2$ graded Poisson structure \cite{BerezinMarinov, Berezin}.

For both the standard quantisation and the non-standard quantisation, and for any member of the hierarchy, we show that the quantum flows can be presented in Heisenberg form:
\[
\partial_{t_{k}}( \bullet )=\frac{1}{\omega^{k}-1}[H^{(k)},\ \bullet]
\]
where the Hamiltonians $H^{(k)}$ are local functionals modulo the associated quantisation ideal, which is defined in Section \ref{sec23}.

 In Section \ref{Sec4}, we extend this approach to a broad class of integrable hierarchies. For the modified Volterra hierarchy, we show that its two distinct non-Abelian versions give rise to identical quantum algebras, while the next symmetry in the hierarchy admits three essentially different quantisations -- one standard and two non-standard. For the Bogoyavlensky hierarchy, we recover the known commutation relations originally derived using the quantum inverse scattering method \cite{InKa}. Finally, we investigate quantisations of several other integrable differential–difference systems, including the Ablowitz–Ladik, relativistic Toda, Merola–Ragnisco–Tu, Adler-Yamilov, Chen–Lee–Liu, Belov–Chaltikian, and Blaszak–Marciniak lattices. For these hierarchies, we obtain new non-standard quantisations as well as bi-quantum structures.

\section{The algebraic quantisation approach}\label{sec2}

In this section, we provide a brief description of the algebraic quantisation approach based on quantisation ideals and quantum algebras,
introducing the basic definitions and notations required.
The details of this approach applied to both finite and infinite dimensional systems can be found in \cite{ BM2021, CMW, CMW2}.

\subsection{Non-Abelian integrable hierarchies}\label{sec21}
Let $\fA=\fieldk\langle u^i_n\,;\, i=1, \cdots, \ell,\ n\in\bbbz\rangle$ be the unital free associative
algebra of polynomials generated by an infinitely many non-commutative
variables $u^i_n$. 
We denote its unit element by $\mathbbm{1}$ and simply write $u^i$ instead of $u^i_0$. This non-Abelian algebra is a difference algebra for the natural \textit{shift} automorphism $\cS$ defined on the generators by 
\begin{equation} \label{ess}
\cS(u_n^{i})=u_{n+1}^{i}, \quad   \cS(\mathbbm{1})=\mathbbm{1}
\end{equation}
and extended multiplicatively to all elements in $\fA$.
For $a \in \fA$ and integer $n \in \mathbb{Z}$, we sometimes denote $\cS^n(a)$ by $a_n$.

The algebra $\fA$ admits an anti-automorphism $\cT$ defined on generators by
\begin{equation} \label{taut}
\cT(u^i_n) = u^i_{-n}, \quad \cT(\mathbbm{1}) = \mathbbm{1},
\end{equation}
and extended to $\fA$ via $\cT(ab) = \cT(b)\cT(a)$ for all $a,b \in \fA$.
Most non-Abelian integrable hierarchies are stable under $\cT$ \cite{NW25}. 

A derivation  $\cD  $ of the algebra  $\fA$ is a $\C$--linear map $\cD : \fA \to \fA$
satisfying the Leibniz's rule:
\[\cD  (\alpha f+\beta  g)=\alpha\cD  (f)+\beta\cD  (g),\quad  \cD  (f\
g)=\cD (f)\ g+f\ \cD (g),\quad  f, g\in\fA,\ \ \alpha, \beta\in\C.\]

We say that $\cD $ is \textit{evolutionary} if it commutes with the shift
operator
$\cS$, in which case it is completely characterised by its action on
the generators $u^i$, that is,
\[
 \cD  (u^i)=X^i\quad \mbox{and} \quad \cD (u^i_k)=\cS^k (X^i),\qquad X^i\in \fA,
\]
where $X=(X^1, \cdots, X^{\ell})\in \fA^{\ell}$ is called the \textit{characteristics} of the derivation.
We denote this derivation by $\cD_X$. 

Evolutionary
derivations form a Lie subalgebra of the Lie algebra of the derivations of $\fA$,
and the characteristic of a commutator $[\cD_X,\cD_Y]=\cD_Z$ is given by
$$Z^i=\cD_X(Y^i)-\cD_Y(X^i),\qquad Z=(Z^1, \cdots, Z^{\ell}).$$
Thus, it induces a Lie bracket in $\fA^{\ell}$.

Assuming that the generators $u^i_k$ depend on a time variable $t\in\C$, we identify the
evolutionary derivation $\cD_X$ with an infinite system of
equations:
\begin{equation}\label{evoeq1}
 \d_t(u^i_n)=\cD_X(u^i_n)=\cS^n(X^i),\quad i=1, \cdots, \ell, \quad X^i\in \fA, \quad n\in\Z.
\end{equation}
From now on, we identify the system of equations and the
evolutionary derivation $\d_t: \fA \mapsto \fA$. System \eqref{evoeq1} can be represented by a single equation with $n=0$:
\begin{equation}\label{evoeq}
 \d_t(u^i)=\cD_X(u^i)=X^i,\quad i=1, \cdots, \ell, \quad X^i\in \fA.
\end{equation}

Evolutionary derivations $\d_{\tau}: \fA \mapsto \fA$ that commute with $\d_{t}$ are called (generalised)
symmetries of the system \eqref{evoeq}, that is,  $[\d_{t},\d_{\tau}]=0$.

We call a (non-Abelian) system \eqref{evoeq} integrable if it is contained in an infinite-dimensional abelian Lie subalgebra of evolutionary derivations of $\fA$. The members of such a Lie subalgebra form an integrable hierarchy.

A well-known example of such an object is the non-Abelian Volterra hierarchy, which consists of an infinite family of commuting evolutionary derivations $(\partial_{t_k})_{k \geq 1}$ of the algebra $\fieldk\langle u_n\,; \,  n\in\bbbz\rangle$.  Its first member is the non-Abelian Volterra equation itself:
\begin{equation}\label{vol}
 u_{t_1}=u_{1} u -u u_{-1} \, \,.
\end{equation}
 The second member of the hierarchy,
\begin{equation}\label{secV}
 u_{t_2}=u_2 u_1 u+u_1^2 u+u_1 u^2-u^2 u_{-1}-u u_{-1}^2 -uu_{-1}u_{-2},
\end{equation}
defines the derivation $\p_{t_2}:\fA\mapsto\fA$ that commutes with $\partial_{t_1}$.

\subsection{Quantisation ideals and quantum algebras}\label{sec22}
\begin{definition}\label{Def1}
For a dynamical system \eqref{evoeq}, i.e., an evolutionary derivation $\d_t: \fA \mapsto \fA$ on a free associative algebra $\fA$,
a two-sided difference ideal
$\fJ\subset\fA$ is called  a {\em quantisation ideal} if it satisfies the following two properties:
\begin{enumerate}
 \item[({\rm i})] the ideal $\fJ$ is $\partial_t$--stable, that is, $\partial_t(\fJ)\subset\fJ$;
 \item[({\rm ii})] the quotient algebra $\fA_\fJ=\fA\diagup\fJ$ admits an additive basis of normally ordered monomials.
\end{enumerate}
The quotient algebra $\fA_\fJ$ is called a {\em quantum algebra}. The induced system $\d_t: \fA_\fJ \mapsto \fA_\fJ$  is referred to as a {\em quantisation} of the dynamical system. Such a quantisation is \textit{standard} if the quantum algebra admits a commutative limit, and \textit{ non-standard} otherwise.
\end{definition}

This definition was first proposed in \cite{AvM20}. The quantisation problem for equation (\ref{evoeq}) reduces
to the problem of finding two-sided ideals such that conditions (i) and (ii) are satisfied.

The condition (i) is crucial. The dynamical system (\ref{evoeq}) defines a derivation
$\partial_t:\fA\mapsto\fA$.  The $\partial_t$--stability of the ideal guarantees that the induced dynamical system on the quantum algebra $\fA_\fJ$ is well defined.

Condition (ii) enables the definition of commutation relations
between any two elements of the quotient algebra and unique representations of
elements of  $\fA_\fJ$ in the basis of normally ordered monomials.

In this paper, we consider quadratic two-sided difference ideals generated by
polynomials of degree two in terms of the generators of $\fA$ and their shifts.

\begin{example}\label{EG1} In the scalar case ($\ell=1$), we have $\fA=\fieldk\langle u_n\,;\ n\in\bbbz\rangle$, where we drop the upper index.  The following family of 
quadratic ideals was considered in \cite{AvM20}:
\begin{equation}\label{ideal00}
 \fI=\langle \ff_{i,j}\,;\, i<j,\ i,j\in\Z\rangle,\qquad
 \, \ff_{i,j}=u_i u_j-\omega_{i,j}u_j u_i,
\end{equation}
where $\omega_{i,j}\in\C^*$ are arbitrary non-zero complex parameters. These ideals clearly satisfy condition (ii).

For the non-Abelian Volterra system (\ref{vol})
it has been shown in \cite{AvM20, CMW} that the ideal $\fI$ satisfies condition (i) if and only if
\[
 \omega_{n,n+1}=\omega,\qquad  \omega_{n,m}=1 \ \
\mbox{if}\ \ |n-m|\geqslant 2.
\]
Thus the only quantisation ideal for the Volterra equation (\ref{vol}) in that class of ideals is
\begin{equation}\label{idi}
\fI_{\omega}= \langle  u_nu_{n+1}-\omega u_{n+1}u_n\,,
u_nu_m-u_mu_n\,;\ |n-m| >1 ,\ n,m \in\bbbz\  \rangle ,
 \end{equation}
implying the commutation relations
\begin{equation*}\label{comm1}
  u_nu_{n+1}=\omega u_{n+1}u_n,\qquad u_nu_m=u_mu_n\ \
\mbox{if}\ \ |n-m|\geqslant 2,\quad n,m \in\bbbz
\end{equation*}
in the quotient algebra $\fA_{\fI_{\omega}}$. This is a standard quantisation: the limit $\omega\to 1$ recovers the classical commutative case and its quadratic Poisson bracket.

The ideal $\fI_{\omega}$ (\ref{idi}) is also stable with respect to $\p_{t_2}$ defined by \eqref{secV} and all members of the Volterra hierarchy \cite{CMW}.
Moreover, the derivation $\p_{t_2}$ admits another stable ideal of type (\ref{ideal00}):
  \begin{equation}\label{idj0}
\hat{\fI}_{\omega}= \langle u_nu_{n+1}-(-1)^n \omega u_{n+1}u_n\,,\,
u_nu_m+u_mu_n\,;\,|n-m| >1,\  n,m \in \mathbb{Z} \rangle\, .
\end{equation}
The quotient algebra $\fA_{{\hat{\fI}}_{\omega}}$ is noncommutative for any choice of the quantisation parameter $\omega$. Thus, it is a non-standard quantisation. 
When treated as a deformation of a non-commutative algebra, it produces a commutative Poisson algebra and a Poisson module in the limit $\omega\to 1$ \cite{MV}.
\end{example}

\subsection{Local functionals}\label{sec23}
\begin{definition}\label{Deflocal}
 Let $\fI$ be a difference ideal in a difference algebra $\fA$.
We say that an element $a \in \fA$ is \textit{local} modulo the ideal $\fI$ if for all $b \in \fA$, there exists an integer $N_b >0$ such that for all $|n|> N_b$, $[a, \cS^n(b)] \in \fI$. We say that a difference ideal $\fI$ is \textit{local} if every element in $\fA$ is local modulo the ideal $\fI$.
\end{definition}
The set of all local elements modulo $\fI$ forms a subalgebra of $\fA$, denoted by $\fA'$. For a given $\fI$, this set can be determined.
\begin{example}
    In Example \ref{EG1}, every element of $\fA$ is local modulo ${\fI}_{\omega}$ defined by \eqref{ideal00} since the commutators $[u_n,\ u_m]$ are in ${\fI}_{\omega}$ when $|n-m| >1$. Thus, ${\fI}_{\omega}$ is a local difference ideal of $\fA$.

For the difference ideal ${\hat{\fI}}_{\omega}$ defined by \eqref{idj0}, it is clear that $ u$ is not a local element modulo this ideal, hence the ideal ${\hat{\fI}}_{\omega}$ is not local. However, some elements are local, e.g., $u^2\in \fA'$.
%More precisely, an element of $\fA$ is local modulo ${\hat{\fI}}_{\omega}$ if and only if it contains only monomials of even degrees.
\end{example}

Starting from Definition \ref{Deflocal}, we define local functionals. Subsequently, the Hamiltonians appearing in the Heisenberg equations are local functionals modulo the associated quantisation ideal.
\begin{definition}
 The elements of the quotient space $\fF_{\fI}=\fA'/(\cS-1) \fA'$ are called \textit{local functionals  modulo $\fI$} and are denoted by $\sum_{n \in \mathbb{Z}} a_n$. Each local functional acts on $\fA_{\fI}$ as an evolutionary derivation written as $[ \sum_{n \in \mathbb{Z}} a_n \, , \, \bullet \, ]$. It is defined by the formula    
 $$ \Big[ \sum_{n \in \mathbb{Z}} a_n \, , \, b \, \Big] := \Big[ \sum_{n=-N_b}^{N_b} a_n \, , \, b \, \Big], $$
  where $N_b$ is chosen such that $[a_n, b] \in \fI$ for all $|n| > N_b$, that is, $a\in \fA$ is local modulo $\fI$.
\end{definition}
The following lemma is used extensively to verify that a difference ideal is a quantisation ideal of a non-Abelian system.

\begin{Lem} \label{keylem}
    Let $\cD$ be an evolutionary derivation of $\fA$, $\fI$ a difference ideal in the algebra $\fA$ and $\sum_{n \in \mathbb{Z}} a_n\in \fF_{\fI}$ a local functional modulo $\fI$. Suppose that for all $1 \leq i \leq \ell$, the following holds in the quotient algebra ${\fA}_{\fI}$:
    \begin{equation} \label{heisgen}
     \cD(u^{i})= \Big[ \sum_{n \in \mathbb{Z}} a_n,\  u^{i} \Big] \,  .
    \end{equation}
    Then, the ideal $\fI$ is $\cD$--stable  and for all $b \in \fA$ we have modulo $\fI$:
    \begin{equation}\label{heis}
     \cD(b)= \Big[\sum_{n \in \mathbb{Z}} a_n,\ b \Big] ,
 \end{equation}
\end{Lem}
\begin{proof}
By the Leibniz rule,  the identity \eqref{heis} follows immediately from \eqref{heisgen}.
\end{proof}
Note that evolutionary derivations of the quotient algebra ${\fA}_{\fI}$ that are generated by local functionals form a Lie subalgebra of the Lie algebra of evolutionary derivations.

\section{Quantisations of the Toda lattice}\label{sec3}
In this section, we first give a short description of the properties of the Toda lattice on a free associative algebra such as its Lax representation, recursion operator, and explicit expressions for its first few generalised symmetries. All such information can be found in the literature, for example \cite{cw19-2} and references therein. After recalling the Poisson pencil for the classical Toda hierarchy,
we proceed to the study of its quantisations. Exploiting the interrelation between the non-Abelian Toda and non-Abelian Volterra hierarchies, we prove that the non-Abelian Toda hierarchy $\d_{t_j}, j\in\bbbn$  admits a standard bi-quantisation. Moreover, we find that its even sub-hierarchy $\d_{t_{2j}}, j\in\bbbn$ 
 admits a non-standard quantisation. In all cases, we also discuss the Heisenberg forms of the quantum hierarchies and present the corresponding Hamiltonians.
\subsection{Toda lattice on free associative algebra}\label{sec31}
\iffalse
The matrix Toda lattice is derived from a discrete version of the principal chiral field model proposed by Polyakov \cite{bmrl80, k81, mik81}.
%The quantisation for periodic Toda lattice was studied in \cite{a}.
Recently, the non-abelian Toda lattice emerged in the study of matrix-valued hermite polynomials \cite{b}.
\fi
The Toda lattice in the Manakov-Flaschka coordinates \cite{Manakov, Flaschka} can be naturally lifted to free associative algebras and is written as
\begin{equation}\label{Toda}
	\begin{cases}
	a_t = b_1 a - a b\\
	b_t = a - a_{-1}
	\end{cases}
\end{equation}
where we use $a$ and $b$ as dependent variables instead of $u^1$ and $u^2$. The associated difference algebra is $\fA= \fieldk\langle a_n, b_m\,;\ n , m \in\bbbz\rangle$.

The Toda lattice admits a scalar Lax representation, namely, there exist two difference operators.
\begin{equation}\label{laxop}
L=\cS+b_1 +a \cS^{-1}\quad \mbox{and} \quad A=L_{+}=\cS+b_1
\end{equation}
such that \eqref{Toda} is equivalent to
\begin{equation}\label{laxT}
L_t=[A,\ L]=A L-L A .
\end{equation}
The Toda lattice \eqref{Toda} is part of a infinite family of pairwise commuting evolutionary derivations $(\partial_{t_k})_{k \in \bbbn}$ called the symmetry hierarchy. It can be generated from the $L$ operator given by (\ref{laxop}) as follows:
\begin{equation}\label{laxh}
L_{t_k}=[A_k,\ L],   \quad A_k= (L^k)_{+}.
\end{equation}
Obviously, we have $\partial_t=\partial_{t_1}$. When $n=2$, it leads to $$A_2=\cS^2 +(b_1+b_2) \cS+b_1^2+a_1+a.$$ Using (\ref{laxh}), we obtain the first commuting symmetry flow of (\ref{Toda}):
\begin{equation}\label{2ndtoda}
\begin{cases}
	a_{t_2} = b_1^2 a +a_1 a -a b^2-a a_{-1} \\
	b_{t_2} = b_1 a +a b-ba_{-1}- a_{-1} b_{-1}
	\end{cases}
\end{equation}
The recursion operator $\Re$ of the Toda system, mapping the characteristics of a symmetry $( \partial_{t_n}(a), \partial_{t_n}(b) )^T$ to the next one for any $n \geq 1$, was constructed in \cite{cw19-2}:
	\begin{equation}
	\Re=\begin{pmatrix}
	\l_{b_1}\r_a\cS-\l_a\r_b & \r_a\cS^2-\l_a\\
	\r_a\cS-\l_{a_{-1}}\cS^{-1} & \r_b \cS-\l_b
	\end{pmatrix}
	\begin{pmatrix}
	\left(\r_a\cS-\l_a\right)^{-1} & 0\\0 & (\cS-1)^{-1}
	\end{pmatrix}.
	\end{equation}
Here $\l_s$ and $\r_s$ denote, respectively, the operator of multiplication by $s$ on the left and on the right. We can directly check that
$
\Re\left( a_t,\ b_t\right)^T=\left(a_{t_2}, \ b_{t_2}\right)^T 
$
and $\Re\left(a_{t_2}, \ b_{t_2}\right)^T $ gives us
\begin{equation}\label{3rdtoda}
\begin{cases}
	a_{t_3} = b_2 a_1 a+b_1^3 a +b_1 a_1 a+b_1 a^2+a_1 b_1 a -a b^3-a^2 b-a a_{-1}b-ab a_{-1}-a a_{-1} b_{-1} \\
	b_{t_3} = b_1^2 a+b_1 a b+a_1 a+a^2 +a b^2-b^2 a_{-1}-a_{-1}^2-b a_{-1} b_{-1}- a_{-1} b_{-1}^2-a_{-1} a_{-2}
	\end{cases}
\end{equation}
The sub-hierarchy of \eqref{laxh} consisting of all the even members $\partial_{t_{2k}}, k\in \bbbn$ is called the Toda even sub-hierarchy.

\subsection{Poisson pencil for the classical Toda hierarchy}\label{sec32}

Consider the abelian Toda hierarchy defined in the commutative difference algebra 
\begin{equation} \label{comalg}
    \mathcal{V}=\mathbb{C}[a_n, b_m \, ; \, n , m \in \mathbb{Z} ].
\end{equation}
 This hierarchy consists of the family of derivations $(\partial^c_{t_k})_{k \geq 1}$ which are the projections of the derivations $(\partial_{t_k})_{k \geq 1}$ to the quotient algebra $\mathcal{V} \cong \fA / \langle [\fA, \fA]\rangle$ over the ideal generated by commutators. This commutative algebra is equipped with a Poisson pencil for which the abelian Toda hierarchy is bi-Hamiltonian. These brackets
$\{ \, , \, \}_{2}$ and $\{ \, , \, \}_{1}$ are given on the generators of $\mathcal{V}$ by 
\begin{equation} \label{poibra}
    \begin{split}
        \{  a_n \, , \, a_m   \}_2&= (\delta_{n+1,m}- \delta_{n-1,m}) \, a_ma_n ;\\
        \{  b_n \, , \, b_m   \}_2&= \delta_{n+1,m} \, a_n -\delta_{n-1,m}a_m; \\
        \{  a_n \, , \, b_m   \}_2&= \delta_{n+1,m} \, b_ma_n -\delta_{n,m} a_n b_n ;\\
        \{  a_n \, , \, b_m   \}_1&= \delta_{n+1,m} a_n \, -\delta_{n,m} a_n,
    \end{split}
\end{equation}
where the brackets between the remaining pairs of generators are trivial and
$\delta_{n,m}$ denotes the Kronecker symbol. These Poisson brackets commute with the shift automorphism $\cS$ and 
%are local \textcolor{red}{local?}
$$ \forall f, g \in \mathcal{V},  \, \exists N \in \mathbb{Z}_+, \forall k \in \mathbb{Z}, |k| > N \implies \{ f, \cS^k(g) \}_{1,2} =0.$$
Due to these properties we can define two evolutionary  derivations for any element $f \in \mathcal{V}$, namely 
$$ \left\{ \sum_{n \in \mathbb{Z}} f_n , \, \bullet \, \right\}_1 \, \, \text{and} \, \,  \left\{ \sum_{n \in \mathbb{Z}} f_n , \, \bullet \, \right\}_2 \, .$$
It turns out that all members in the classical hierarchy $(\partial^c_{t_k})_{k \geq 1}$ can be expressed in two ways. Namely, 
\begin{equation*}
 \partial^c_{t_k} (\bullet)=-\left\{ \ \sum_{n\in\Z} \cS^n \rho^{(k-1)} ,\, \bullet \right\}_2=- \left\{ \ \sum_{n\in\Z} \cS^n \rho^{(k)} ,\, \bullet \right\}_1 ,
\end{equation*}
where 
\begin{equation*}
    \rho^{(0)}=b ,\qquad \quad  \rho^{(1)}=  \frac{b^2}{2}+a, \qquad 
    \rho^{(2)}= \frac{b^3}{3}+a(b+b_1), \qquad \ldots\ .
\end{equation*}

\subsection{Bi-quantisation for the non-Abelian Toda Hierarchy }\label{sec33}
It is known that there is a link between the Toda lattice and the Volterra lattice.
In this section we discuss how to connect the results on the quantisation of the non-Abelian Volterra lattice \cite{CMW, CMW2}
with the non-Abelian Toda hierarchy.

For convenience, we denote $\fB=\fieldk\langle u_n\,;\ n\in\bbbz\rangle$. The non-Abelian Volterra lattice (\ref{vol}) is a member of the hierarchy of commuting evolutionary derivations
$\bd_{t_k}: \fB \mapsto \fB$ on $\fB$ defined via the Lax operator $M$ as follows:
\begin{equation} \label{laxvolt}
 M= \cS + u \cS^{-1} \quad \mbox{and} \quad   \bd_{t_k}(M)= [ (M^{2k})_+, M] , \, \, \, k \in\bbbn.
\end{equation}
Here we use $\bd$ instead of $\d$ to distinguish the derivations on $\fB$ from the ones on $\fA$.
In particular, $\bd_{t_1}$ corresponds to the non-Abelian Volterra lattice (\ref{vol}).

Recall that $\fA=\fieldk\langle a_n, b_n\,;\ n\in\bbbz\rangle$. Let $\gamma \in \mathbb{C}$ and $\phi_{\gamma} : \fA \mapsto \fB$ be the injective morphism of algebras defined by
\begin{equation}\label{morphism}
 \phi_{\gamma}(a_n)=u_{2n}u_{2n-1} \, , \, \, \phi_{\gamma}(b_n)=u_{2n-1}+u_{2n-2}+ \gamma \, ,
\end{equation}
where we use the notation $\gamma \mathbbm{1}=\gamma$.

Note that $\phi_{\gamma}$ is a morphism of difference algebras from $(\fA,\ \cS)$ to $(\fB, \ \cS^2 )$.
Following from (\ref{laxvolt}), we have for all $n\geq 1$,
\begin{equation} \label{laxvolt2}
    \bd_{t_k}(M^2)= [ (M^{2k})_+,\ M^2] , \, \, \, k \in \bbbn
\end{equation}
where
\begin{equation}\label{m2}
  M^2=\cS^2+(u_1+u)+uu_{-1}\cS^{-2}= \cS^2+\phi_{\gamma}(b_1)-\gamma+ \phi_{\gamma}(a) \cS^{-2} .
\end{equation}
Comparing to the Lax operator for the Toda lattice (\ref{laxop}), it follows that $\phi_{\gamma}$ is an intertwiner between both non-Abelian hierarchies:
\begin{equation}\label{morel}
\sum_{j=1}^k (-\gamma)^{k-j} \binom{k}{j} \phi_{\gamma} \partial_{t_j} = \bd_{t_k} \phi_{\gamma} \, \, \, \, \text{for all } k \geq 1 \, .
\end{equation}

\begin{Lem}\label{lemma1}
    Let $\phi_{\gamma}$ be the morphism defined in \eqref{morphism} and let $\fK$ be a difference ideal of $\fB$. If $\fK$ is stable under the derivations $\bd_{t_k}$, $k \in \bbbn$, of
    the non-Abelian Volterra hierarchy, then its preimage $\phi_{\gamma}^{-1}(\fK)$ is a stable difference ideal of $\fA$ under the derivations $\d_{t_k}$ of the non-Abelian Toda hierarchy.
\end{Lem}
\begin{proof}
    Let $f \in \phi_{\gamma}^{-1}(\fK)$ and $k \in\bbbn$. Since $\fK$ is stable for the derivation $\bd_{t_k}$ and $\phi_{\gamma}(f) \in \fK$, we have $\bd_{t_k}(\phi_{\gamma} (f)) \in \fK$ for all $k\in\bbbn$. After solving a triangular system obtained from \eqref{morel}, we see that $\phi_{\gamma}(\partial_{t_k}(f)) \in \fK$ for all $k\in\bbbn$, which means that $\phi_{\gamma}^{-1}(\fK)$ is stable for the derivations $\partial_{t_k}$.
\end{proof}
We have shown in \cite{CMW} that the ideal $\fI_{\omega}$ given by (\ref{idi})
is stable for the derivations in the non-Abelian Volterra hierarchy. Using this lemma, we obtain the following statement on the non-Abelian Toda hierarchy:
\begin{The}\label{thm1}
\iffalse
Let $\omega$ be a arbitrary nonzero complex number. The difference ideal $\fK_\omega\subset \fA$ generated by the set of polynomials
 \begin{equation}\label{id1}
 \begin{array}{lll}
 b_nb_{n+1}-b_{n+1}b_n-(\omega-1)\, a_n,\qquad
&b_nb_m-b_mb_n,\quad &
|n-m|\ne  1, \\
 a_na_{n+1}-\omega a_{n+1} a_n,\qquad &a_na_m-a_ma_n,\quad &
|n-m|\ne  1, \\
 b_na_n-\omega a_nb_n,&a_nb_m-b_ma_n,&m-n\ne0,1,\\
    a_nb_{n+1}-\omega b_{n+1}a_n, & m,n \in \bbbz& 
 \end{array}
\end{equation}
 is stable with respect to all the derivations $\partial_{t_n}$ of the non-Abelian Toda hierarchy.
 \fi
 Let $\omega$ and $\eta$ be arbitrary complex numbers with $\omega \neq 0$. The ideal $\fJ(\omega,\eta)\subset \fA$ generated by the following set of polynomials
 \begin{equation}\label{id11}
 \begin{array}{lll}
 b_nb_{n+1}-b_{n+1}b_n-(\omega-1)\, a_n,\qquad
&b_nb_m-b_mb_n,\quad &
|n-m|\ne  1, \\
 a_na_{n+1}-\omega a_{n+1} a_n,\qquad &a_na_m-a_ma_n,\quad &
|n-m|\ne  1, \\
 b_na_n-\omega a_nb_n-\eta a_n,&a_nb_m-b_ma_n,&m-n\ne0,1,\\
    a_nb_{n+1}-\omega b_{n+1}a_n-\eta a_n, & m,n \in \bbbz& 
 \end{array}
\end{equation}
 is stable with respect to the derivations $\partial_{t_k}$ of the non-Abelian Toda hierarchy (\ref{laxh}).
\end{The}
\begin{proof}
By a generic argument, it is enough to show this statement when $\eta=\gamma(1-\omega)$. We use Lemma \ref{lemma1} and check that $\fJ(\omega, \gamma(1-\omega))=\phi_{\gamma}^{-1}(\fI_{\omega})$, where the ideal $\fI_{\omega}$ is given by (\ref{idi}).
The inclusion $\fJ(\omega, \gamma(1-\omega))\subset \phi_{\gamma}^{-1}(\fI_{\omega})$ can be checked by direct computation. We show it explicitly for the first defining relation of $\fJ(\omega, \gamma(1-\omega))$: 
\begin{eqnarray*}
 &&\phi_{\gamma}\left( b_nb_{n+1}-b_{n+1}b_n-(\omega-1)\, a_n\right) = \phi_{\gamma}( b_n) \phi_{\gamma}(b_{n+1})-\phi_{\gamma}(b_{n+1})\phi_{\gamma}(b_n)-(\omega-1)\, \phi_{\gamma}(a_n)\\
 &&\qquad=(u_{2n-1}+u_{2n-2}+\gamma) (u_{2n+1}+u_{2n}+\gamma) \\
 &&\qquad\qquad- (u_{2n+1}+u_{2n}+\gamma) (u_{2n-1}+u_{2n-2}+\gamma)-(\omega-1) u_{2n}u_{2n-1}\\
 &&\qquad  = u_{2n-1} u_{2n}-u_{2n} u_{2n-1} -(\omega-1) u_{2n}u_{2n-1} \mod \fI_{\omega}
 \\
 &&\qquad=0 \mod \fI_{\omega} .
\end{eqnarray*}
We can then define a quotient map $\phi_{\omega, \gamma} :\fA_{\fJ(\omega, \gamma(1-\omega)) }\rightarrow \fB_{\fI_{\omega}}$. Note that this quotient map is injective if and only if $\fJ(\omega,\gamma(1-\omega))=\phi_{\gamma}^{-1}(\fI_{\omega})$. Both quotient algebras have a basis of ordered monomials defined from the total orders $ ... > a_n > b_n > a_{n-1} > ... $ and $ ... u_n > u_{n-1} ...$. Let $f$ be a nonzero polynomial in $\fA_{\fJ(\omega, \gamma(1-\omega))}$ and ${\rm Lm}(f)$ be its leading term, which means the monomial with highest lexicographical order. The image of ${\rm Lm}(f)$
 in $\fB_{ \fI_{\omega}}$ is non-trivial. Indeed the leading term in the image is uniquely determined by ${\rm Lm}(f)$:
 $$ {\rm Lm}(f)=c\prod_{j}a_{j}^{n_j}b_j^{l_j} \rightarrow  c\prod_{j}u_{2j}^{n_j}u_{2j-1}^{n_j+l_j} = \text{leading term of } \phi_{\omega, \gamma}(f).$$
 Therefore $\phi_{\omega, \gamma}$ is injective and $\fJ(\omega, \gamma(1-\omega))=\phi_{\gamma}^{-1}(\fI_{\omega})$.
\end{proof}

 Let $\fA_{\omega , \eta}$ be the quotient algebra $\fA / \fJ(\omega, \eta)$.
 %$\left(\fA \otimes \bbbc [[\omega]] \otimes \bbbc [ \eta ]\right)/ \fJ(\omega, \eta)$.
 Suppose that $\omega=1+ \mu \alpha$, $\eta=\mu \beta$, where $\alpha, \beta \in \mathbb{C}$ and $\mu$ is an indeterminate. As a vector space, we can identify $\fA_{\omega, \eta}= \mathcal{V}[[\mu]]$,
 where $\mathcal{V}$ is defined in \eqref{comalg}.  For all $f, g \in \fA_{\omega, \eta}$, we have
    \begin{equation*}
        fg-gf=\mu \Big( \alpha \{ \, \pi(f) \, , \, \pi(g) \, \}_2 + \beta \{ \, \pi(f) \, , \, \pi(g) \, \}_1 \Big) + o(\mu),
    \end{equation*}
where $\pi(\sum_{l \geq 0} v_{(l)} \mu^l) := v_{(0)}$ and the Poisson brackets are defined in \eqref{poibra}. Therefore, $\fI(\omega,\eta)$ is a standard bi-quantisation ideal.

The ideal $\fJ(\omega, \eta)$ is local by Definition \ref{Deflocal} in Section \ref{sec23}. So we can construct an evolutionary derivation of $\fA_{\omega, \eta}$ from any local functional.
In our previous \cite{CMW2}, we constructed an infinite sequence of Hamiltonians for the quantisation of the non-Abelian Volterra hierarchy with respect to the ideal $\fJ_{\omega}$.
These can be pulled back through $\phi_{0}$ to produce Hamiltonians for the quantum Toda hierarchy with respect to the ideal $\fJ(\omega,\eta)$. 
The first three Hamiltonians in $\fB_{\fI_{\omega}}$ for the quantum Volterra lattice are written as:
\begin{eqnarray}
&&H_1=\sum_{k\in\mathbb{Z}}{u_k}; \qquad H_2=\frac{1}{2} \sum_{k\in\mathbb{Z}} \left({u_k^2}+(\om+1) {u_{k+1} u_k}\right);\label{volH1}\\
&&H_3= \frac{1}{3} \sum_{k\in\mathbb{Z}} \left({u_k^3}+(\om^2+\om+1) (u_{k+2}u_{k+1}u_k+u_{k+1}^2u_k+u_{k+1}u_k^2)\right).\nonumber
\end{eqnarray}
Through the morphism $\phi_0$ \eqref{morphism}, we obtain the first three Hamiltonians for the quantum Toda hierarchy with respect to the ideal $\fJ(\omega,\eta)$ as follows:
\begin{equation*}
    \begin{split}
        H_1 &= \sum_{n \in \mathbb{Z}} b_n \\
        H_2 &= \frac{1}{2}\sum_{n \in \mathbb{Z} }\left( 2a_n + b_n^2+b_n b_{n+1} - b_{n+1}b_n\right) \\
        H_3&= \frac{1}{3}\sum_{n \in \mathbb{Z}} \left(2 a_nb_n+2a_nb_{n+1}+b_na_n+b_{n+1}a_n+b_n^3+b_n^2b_{n+1} +b_nb_{n+1}^2-b_{n+1}b_n^2-b_{n+1}b_nb_{n+1}\right)
    \end{split}
\end{equation*}
These Hamiltonians enable us to write the quantum Toda hierarchy in Heisenberg form with quantum coefficients. Explicitly, for the first three derivations, we have 
\begin{equation} \label{Hampen1}
    \begin{split}
        \big[ H_1 \, , \, \bullet \, \big] &= (1-\omega) \partial_{t_1}(\bullet), \\
        \big[ H_2 \, , \, \bullet \, \big] &= \frac{1-\omega^2}{2} \partial_{t_2}(\bullet)-\eta \omega \partial_{t_1}(\bullet),  \\
        \big[ H_3 \, , \, \bullet \, \big]&= \frac{1-\omega^3}{3}\partial_{t_3}(\bullet)-\eta \omega^2 \partial_{t_2}(\bullet)-\eta^2 \omega  \partial_{t_1}(\bullet).
    \end{split}
\end{equation}   
\subsection{Non-standard quantisation of the non-Abelian Toda sub-hierarchy}\label{sec34}
We have shown in \cite{CMW} that the even sub-hierarchy $(\bd_{t_{2k}})_{k \geq 1}$ of the non-Abelian Volterra equation not only stabilises the ideal $\fJ_{\omega}$ given by (\ref{idi}) but also the quantisation ideal $\hat{\fI}_{\omega}$ given by \eqref{idj0}. 
This quantisation is non-standard.
Using Lemma \ref{lemma1} we have a similar picture for the non-Abelian Toda hierarchy obtained by pulling back this ideal through the morphism $\phi_0$ \eqref{morphism} as follows:

\begin{The}
Let $\omega$ be an arbitrary nonzero complex number. The ideal $\hat{\fK}_{\omega} \subset \fA$ generated by the relations
 \begin{equation}
 \begin{array}{rclll}
 \phantom{e^{-i\hbar}} b_nb_{n+1}+b_{n+1}b_n&=&(1- \omega)\, a_n,\qquad
&b_nb_m+b_mb_n=0,\quad &
|n-m|\ne  1, \\
 a_na_{n+1} + \omega a_{n+1} a_n&=&0,\qquad &a_na_m-a_ma_n=0,\quad &
|n-m|\ne  1, \\
  b_na_n + \omega a_nb_n&=& 0,&a_nb_m-b_ma_n=0,&m-n\ne0,1,\\
    a_nb_{n+1}+ \omega b_{n+1}a_n&=&0,& m,n \in \bbbz&
 \end{array}
\end{equation}
 is stable with respect to all members in the non-Abelian Toda even sub-hierarchy, i.e., $\partial_{t_{2k}}, k\in \bbbn$ defined by \eqref{laxh}.
\end{The}
\begin{proof}
We prove the statement in the same way as for Theorem \ref{thm1}. Here we verify the inclusion $\hat{\fK}_{\omega} \subset \phi_0^{-1}(\hat{\fI}_{\omega})$, where the ideal $\hat{\fI}_{\omega}$ is given by (\ref{idj0}).
For the first relation, we have
\begin{eqnarray*}
 && \phi_0\left( b_nb_{n+1}+b_{n+1}b_n-(1-\omega)\, a_n\right)= \phi_0( b_n) \phi_0(b_{n+1})+\phi_0(b_{n+1})\phi_0(b_n)-(1-\omega)\, \phi_0(a_n)\\
 &&\qquad=(u_{2n-1}+u_{2n-2}) (u_{2n+1}+u_{2n})+(u_{2n+1}+u_{2n}) (u_{2n-1}+u_{2n-2})-(1-\omega) u_{2n}u_{2n-1}\\
 &&\qquad=u_{2n-1} u_{2n} +u_{2n} u_{2n-1}-(1-\omega) u_{2n}u_{2n-1} \mod \hat{\fI}_{\omega}\\
 &&\qquad=0 \mod \hat{\fI}_{\omega} .
\end{eqnarray*}
All other relations can be checked similarly by direct computation. Finally, we conclude that $\hat{\fK}_{\omega}=\phi_0^{-1}(\hat{\fI}_{\omega})$ since the induced quotient map is injective by the same argument as used in the proof of Theorem \ref{thm1}. The statement then follows from Lemma \ref{lemma1}.
\end{proof}

 By pulling back the Hamiltonians for the non-standard quantisation of the non-Abelian Volterra sub-hierarchy through $\phi_0$, one obtains an infinite family of Hamiltonians for the non-standard quantisation of the non-Abelian Toda sub-hierarchy.
 
 The first two Hamiltonians in $\fB_{\hat{\fI}_{\omega}}$ for the non-standard quantisation of the Volterra sub-hierarchy $\bd_{t_{2k}}$ \cite{CMW2} are:
\begin{eqnarray*}
 &&\hat{H}_2 =  \sum_{n\in\mathbb{Z}} {u_n^2}+
(1+(-1)^{n}\omega){u_{n+1}u_{n}};\\
&&\hat{H}_4 = \sum_{n\in\mathbb{Z}}{u_n^4} -(\om^4-1)
 {u_{n+2}u_{n+1}^2u_n} +(\om^2+1)
 {(\omega^2+(-1)^{n}\omega+1) u_{n+1}^2u_n^2}\\
&&\qquad + (\om^2+1)  (1+(-1)^{n}\om)
\left(u_{n+3}u_{n+2}u_{n+1}u_n+u_{n+2}^2u_{n+1}u_n+u_{n+1}^3u_n+u_{n+1}u_n^3\right)\\
&&\qquad + (\om^2+1)  (1-(-1)^{n}\om) u_{n+2} u_{n+1} u_{n}^2.
\end{eqnarray*}
It leads to the first two Hamiltonians for the non-standard quantisation of the Toda even sub-hierarchy $\d_{t_{2k}}$ using the pull back $\phi_0$:
\begin{eqnarray*}
 &&\hat{H}_2 =  \sum_{m\in\mathbb{Z}}(b_m^2+(1-\om) a_m)=\sum_{m \in \mathbb{Z}} \left(b_m^2+b_m b_{m+1}+b_{m+1}b_m \right);\\
&&\hat{H}_4 = \sum_{m\in\mathbb{Z}}\left(b_m^4+(2-\omega)(1+\omega^2) a_m^2+(1+\omega)(1+\om^2)b_{m+1}a_m b_m\right.
\\ 
&&\left. \qquad \qquad+(1-\om)(1+\om^2)  (a_{m+1}a_m+a_m b_m^2+b_{m+1}^2 a_m)\right) \ .
\end{eqnarray*}
We write the quantum Toda even sub-hierarchy in Heisenberg form using the above Hamiltonians. Explicitly, for the first two members
$\partial_{t_2}$ and $\partial_{t_4}$, we have 
\begin{equation*}
        \big[\hat{H}_2  \, , \, \bullet \, \big] = (1-\omega^2) \partial_{t_2}(\bullet); \qquad 
        \big[\hat{H}_4  \, , \, \bullet \, \big] = (1-\omega^4) \partial_{t_4}(\bullet). 
\end{equation*}

\section{Quantisations of integrable differential-difference systems}\label{Sec4}
In this section, we apply the algebraic quantisation approach described in Section \ref{sec22} to several known non-Abelian integrable systems mainly taken from the list in \cite{cw19-2}.
These systems serve as illustrative examples demonstrating the effectiveness of the algebraic approach; we make no claim to completeness. Our focus is on quadratic ideals, which are generated by quadratic polynomials in the generators of difference algebras as in \eqref{ideal00}. For each system, we present its quantisation ideals and the corresponding Heisenberg forms. The stability of all ideals considered here is ensured by Lemma \ref{keylem}.

\subsection{The modified Volterra lattice}\label{sec41}
The modified Volterra equations possesses two non-Abelian integrable lifts \cite{adler20}:
\begin{equation}\label{mVol1}
v_{t_1}=v_1v^2 -v^2 v_{-1}
\tag{mV1-1}
\end{equation}
and
\begin{equation}\label{mVol2}
v_{\tau_1}=v ( v_1-v_{-1} )v
\tag{mV2-1}
\end{equation}
\iffalse
They both admit a zero order symmetry
\[
 u_{k,t_0}=u_{k,\tau_0}=(-1)^k u_k,
\]
corresponding to the point symmetry $u_k\to \Gamma^{(-1)^k}u_k$. 
\fi
Between them, there is no direct Miura transformation.
The next symmetries of the non-Abelian modified Volterra equations \eqref{mVol1} and \eqref{mVol2}
are given, respectively, by:
\begin{subequations}
\begin{align}
v_{t_2}&=v_2 v_1^2 v^2+v_1 v v_1 v^2+v_1 v^2 v_{-1} v-v v_1 v^2 v_{-1}-v^2 v_{-1} v v_{-1}-v^2 v_{-1}^2 v_{-2} \tag{mV1-2}\label{m2V1}\\
v_{\tau_2}&=v \left( v_1 v_2 v_1 + v_1 v v_1-v_{-1} v v_{-1}-v_{-1} v_{-2} v_{-1}\right)v  \tag{mV2-2}\label{m2V2}
\end{align}
\end{subequations}
Here we use $\partial_{t_k}$ and $\partial_{\tau_k}$, $k\in\bbbn$ to present their hierarchies, respectively.

The substitution $u = v_1 v$ relates the systems \eqref{mVol1} and \eqref{m2V1} to the non-Abelian Volterra equation \eqref{vol} and its first symmetry \eqref{secV}.
\iffalse
\begin{eqnarray}
&&u_{t_1} = u_1 u - u u_{-1},\label{Vw}\\
&&u_{t_2} = u_2 u_1 u + u_1^2 u + u_1 u^2 - u^2 u_{-1} - u u_{-1}^2 - u u_{-1} u_{-2}.\label{Vw2}
\end{eqnarray}
\fi
For a detailed discussion of the connections between the Volterra hierarchy and equations \eqref{mVol1} and \eqref{mVol2}, we refer the reader to \cite{adler20}.

As candidates for a quantisation ideal,  we consider ideals of the form
\begin{equation}\label{JmV}
 J=\langle\, v_m v_n-\omega_{nm} v_n v_m \,|\, n>m\,\rangle,
\end{equation}
where the parameters are $\omega_{nm}\in\C^* $.  These ideals satisfy condition (ii), which means that the quotient algebras $\fA_J$ admit a basis of normally ordered monomials
\begin{equation}\label{pbwbasu}
 \cB=\{ v_{n_1}^{\alpha_1}v_{n_2}^{\alpha_2}\cdots v_{n_k}^{\alpha_k}\,|\, k\in\N,\ n_1>n_2>\cdots >n_k,\ (\alpha_1,\alpha_2,\ldots,\alpha_k)\in \Z_{\geqslant 0}^k\,\}.
\end{equation}
Condition (i) of $\p_{t_1}$--stability imposes constraints on these parameters.
\begin{prop}\label{Jt1mV}
An ideal  of the form (\ref{JmV}) is $\p_{t_1}$--stable if and only if the parameters satisfy 
\begin{equation}\label{JmVL}
  \omega_{nm}=\omega^{(-1)^{n-m+1}}=\left\{\begin{array}{ll}
\omega, \quad &\mbox{if}\ \ n-m\ \ \mbox{is odd}\\
\omega^{-1}, \quad &\mbox{if}\ \ n-m\ \ \mbox{is even}
                                    \end{array}\right. 
\end{equation}
where $\omega\ne 0$ is an arbitrary parameter of quantisation.

An ideal of the form (\ref{JmV}) is $\p_{\tau_1}$--stable if and only if it is $\p_{t_1}$--stable.
\end{prop}
Thus, the quantum algebras for the modified Volterra equations (\ref{mVol1}) and (\ref{mVol2}) coincide. We denote the quantisation ideal by
\begin{equation}\label{J1t1}
 J_1=\langle\, v_m v_n-\omega^{(-1)^{n-m+1}} v_n v_m;\quad n>m\,\rangle.
\end{equation}
It is evident that $J_1$ is $\cS$ invariant. We recall that the second member of the non-Abelian  Volterra hierarchy possesses more stable ideals than the Volterra equation itself.
This phenomenon similarly happens for the non-Abelian modified Volterra hierarchies.
\begin{prop}\label{Jt2mV}
 An ideal of the form (\ref{JmV}) is $\p_{t_2}$--stable if and only if it coincides with $J_1$ (\ref{J1t1}),  or  one of the following three ideals:
 \begin{eqnarray}
J_2&=&\langle v_m v_n-(-1)^{n-m+1}\ \omega^{(-1)^{n-m+1}} v_n v_m;\, n>m\,\rangle,\\
J_3&=&\langle\, v_m v_n-(-1)^{n(n-m+1)}\omega^{(-1)^{n-m+1}} v_n v_m;\, n>m\,\rangle,\\
J_4&=&\langle v_m v_n-(-1)^{(n+1)(n-m+1)}\omega^{(-1)^{n-m+1}} v_n v_m;, n>m\,\rangle,
 \end{eqnarray}
where  $\omega\ne 0$ is an arbitrary parameter.

An ideal of the form (\ref{JmV}) is $\p_{\tau_2}$--stable if and only if it is $\p_{t_2}$--stable.
\end{prop}

The ideal $J_2$ is also $\cS$--invariant while the ideals $J_3$ and $J_4$ are only $\cS^2$--invariant. Note that $J_4=\cS(J_3)$.
Let the quotient algebras $\fA_j=\fA/J_j, j=1, \cdots, 4.$
The algebras $\fA_3$ and $\fA_4$ are isomorphic. The algebras $\fA_2$, $\fA_3$ and $\fA_4$ are not commutative for any choice of the quantisation parameter $\omega$. 
 Similarly to the Volterra hierarchy, $J_1$ is a standard quantisation ideal of the whole non-Abelian modified Volterra hierarchy while $J_2, J_3, J_4$ are non-standard quantisation ideals of its even sub-hierarchy.
These quantisation ideals are related to the ones for the non-Abelian Volterra hierarchy as follows:
\begin{prop}\label{J1234}
Let $\psi$ be the difference algebra morphism between algebra $\fieldk\langle u_n\,;\ n\in\bbbz\rangle$ and algebra $\fieldk\langle v_n\,;\ n\in\bbbz\rangle$, defined by $\psi(u_n)=v_{n+1}v_n$. Then the quantisation ideals in both sides are related as follows. 
\begin{equation}\label{VmV}
        \psi^{-1}(J_1)=\fI_{\omega} \, , \, \, \,  \psi^{-1}(J_2)=\fI_{-\omega} \, ,  \, \, \psi^{-1}(J_3)=\hat{\fI}_{\omega} \, , \, \, \,  \psi^{-1}(J_4)=\hat{\fI}_{-\omega} \,,
\end{equation}
where $\fI_{\omega}$ and $\hat{\fI}_{\omega}$ are given by \eqref{idi} and \eqref{idj0}, respectively.
\end{prop}

Since there is a Miura transformation between the systems \eqref{mVol1} and the Volterra equation, the Hamiltonians for the quantum modified Volterra
equation and its symmetry are closely related to the Hamiltonians for the Volterra equation given by \eqref{volH1}. Indeed, 
letting $u_n=v_{n+1}v_n$, we have the following Hamiltonians
\begin{eqnarray*}
 H^{(1)}&=&\sum\limits_{n\in\Z} u_n=\sum\limits_{n\in\Z} v_{n+1} v_n,\\
  H^{(2)}&=&\sum\limits_{n\in\Z} \left(u_n^2 +u_{n+1}u_n+u_nu_{n+1}\right).
\end{eqnarray*}
Then equations \eqref{mVol1} and \eqref{mVol2} in the algebra $\fA_1$ can be represented in Heisenberg form as follows:
\[
 (1-\omega) \,  \partial_{t_1}(\bullet)=(\omega^{-1}-1) \, \partial_{\tau_1}(\bullet)=\big{[}H^{(1)}, \  \bullet \, \big{]}.
\]
Equation \eqref{m2V1} in the algebras $\fA_j,\ j=1,2,3,4$ can be represented as 
\[
 (1-\omega^2) \partial_{t_2}(\bullet)=\big{[}H^{(2)},\  \, \bullet \, \big{]}.
\]
Finally, in the algebras $\fA_j,\ j=1,2,3,4$ \eqref{m2V2} can be represented as 
\[
 (\omega^{-2}-1) \partial_{\tau_2}(\bullet)=\big{[}H^{(2,j)},\  \bullet \, \big{]}, \quad j=1,2,3,4,
\]
where 
\begin{eqnarray*} H^{(2,1)} &=& H^{(2)}, \quad u_n=v_{n+1} v_n,\\
{H}^{(2,2)}&=&\sum\limits_{n\in\Z} \left(u_n^2 -u_{n+1}u_n-u_nu_{n+1}\right),
\\
{H}^{(2,3)}&=&\sum\limits_{n\in\Z} \left(u_n^2 +(-1)^n (u_{n+1}u_n+u_nu_{n+1})\right),
\\
{H}^{(2,4)}&=&\sum\limits_{n\in\Z} \left(u_n^2 +(-1)^{n+1} (u_{n+1}u_n+u_nu_{n+1})\right).
\end{eqnarray*}
Note that the standard quantisation gives the same quantum equation for both lifts. In fact, in the algebra $\fA_1$, we have 
$$
 (1-\omega^2) \partial_{t_2}(\bullet)=(\omega^{-2}-1) \partial_{\tau_2}(\bullet)=\big{[}H^{(2)},\  \bullet \, \big{]}.
$$
For the non-standard quantisation ideals $J_2, J_3$ and $J_4$, each yields two different quantum system with distinct Hamiltonians.

\subsection{The Bogoyavlensky hierarchies}
In 1991, Bogoyavlensky constructed integrable extensions of the Volterra system in free
associative algebras, one of which is the so-called non-Abelian Bogoyavlensky lattice \cite{Bog91}
\begin{equation}\label{nbogadd}
	u_t=\sum_{j=1}^p \left( u_{j} u- u u_{-j}\right), \quad p\in\mathbb{N}.
\end{equation}
When $p=1$, it is the non-Abelian Volterra lattice \eqref{vol}.

The explicit formulas of its generalised symmetries $Q^{(l)}$, $l\geq 1$, are given in \cite{cw19-2} in terms of a family of non-commutative homogeneous difference polynomials, which is inspired by
the homogeneous difference polynomials in the Abelian case \cite{svin11, wang12}, that is,
\begin{equation}\label{nsymsB}
u_{t_l}=Q^{(l)}=(\r_u \cS-\l_u \cS^{-p})Y^{(l)}=Y^{(l)}_1 u-u Y^{(l)}_{-p},
\end{equation}
where
\begin{equation}\label{yl}
 Y^{(l)}= \sum_{0\leq \lambda_{1}\leq \cdots \leq \lambda_{l}\leq lp-1}
	\left(\prod_{j=1}^{\rightarrow{l}} u_{\lambda_{j}+(1-j)p} \right).
\end{equation}
Here $\prod_{j=1}^{\rightarrow{l}}$ denotes the order of the values $j$, from $1$ to $l$ in the product of the non-commutative
generators $u_{\lambda_j +(1-j)p}$. 
Note that $Y^{(1)}=\sum_{k=0}^{p-1} u_k$ leads to the equation itself, i.e., $u_t=u_{t_1}$.

Mikhailov in \cite{AvM20} showed that equation \eqref{nbogadd} admits a quantisation ideal
 \begin{eqnarray}\label{nibq}
 \fI&= &\Big{\langle}  u_nu_{n+i}-\omega u_{n+i}u_n,\ u_nu_m-u_mu_n;\ 1 \leq i \leq p,\ |n-m| >p \Big{\rangle},
 \end{eqnarray}
 where $\omega$ is a nonzero complex constant. This standard ideal is consistent with the case where $p=1$ is given by \eqref{idi}.
 The above commutation relations, as well as the first few Hamiltonians of the quantum Bogoyavlensky lattice were found using the ultralocal Lax representation and the $r$--matrix technique by Inoue and Hikami \cite{InKa}. The algebraic quantisation approach does not rely on the existence of Lax or Hamiltonian structures.

 In fact, the ideal \eqref{nibq} is stable with respect to all symmetries of the non-Abelian Bogoyavlensky hierarchy, which was first conjectured by Mikhailov
 based on concrete computations for its lower order symmetries.
\begin{The}\label{BOGth}
 All symmetries of the non-Abelian Bogoyavlensky lattice \eqref{nbogadd} when $p\geq 1$ admit a standard quantisation ideal
 $\fI$ given by \eqref{nibq}.
\end{The}
The proof relies heavily on the explicit expressions for the symmetry hierarchy (\ref{nsymsB}) in the same way
as our proof in \cite{CMW} for the Volterra hierarchy.
It involves first proving some lemmas and propositions and thus we collect them in the Appendix
for better presentation flow.

\subsection{The Ablowitz-Ladik lattice}
	
\iffalse
	Its Lax representation is
	\begin{eqnarray*}
		U= \left( \begin{array}{cc} \lambda & u \\ v & \lambda^{-1} \end{array}\right); \qquad
		B=\alpha \left( \begin{array}{cc}\lambda^2-u v_{-1} &  \lambda u  \\ \lambda v_{-1} & 0 \end{array}\right)+\beta \left( \begin{array}{cc}0 &  \lambda^{-1} u_{-1}  \\ \lambda^{-1} v & \lambda^{-2}-vu_{-1} \end{array}\right)
	\end{eqnarray*}
\fi
The classical Ablowitz-Ladik lattice was proposed in \cite{AL76}. Its non-Abelian version 
 \begin{equation}
		\begin{cases} u_t=\alpha(u_1-u_1vu)+\beta(uvu_{-1}-u_{-1}) \\ v_t= \alpha(v u v_{-1}-v_{-1})+\beta(v_1-v_1uv) \end{cases}\quad\alpha,\beta\in\C
	\end{equation}
shares similar algebraic properties such as the existence of a recursion operator generating half of its symmetry hierarchy
$\d_{t^+_k}$  starting from $\d_{t_0}$ defined below. Moreover, the inverse of the recursion operator generates
the second half of its symmetry hierarchy $\d_{t^-_k}$  also starting from $\d_{t_0}$ \cite{cw19-2}. Here are the first few members in this non-Abelian hierarchy:
\begin{eqnarray*}
&&
\begin{cases}
u_{t_0}=u\\
v_{t_0}=-v
\end{cases}\label{ut0} \\&&\nonumber \\
&&\begin{cases}
u_{t^+_1}=u_1 - u_1vu\\
v_{t^+_1}=vuv_{-1} - v_{-1}
\end{cases} \label{ut1p}\\&&\nonumber \\
&&\begin{cases}
u_{t^-_1}=
u_{-1}-uvu_{-1} \\
v_{t^-_1}=
v_1uv-v_1 \end{cases}\label{ut1m}
\\&&\nonumber  \\
&&\begin{cases}
u_{t^+_2}=u_2    (1-v_1 u_1   )    (1-v    u  )+u_1    (v   u
-1   ) v_{-1} u  +u_1 v   u_1    (v   u  -1   )
\\
v_{t^+_2}=   (1-v   u     )    (v_{-1} u_{-1}-1   ) v_{-2}+   (1-v
u     ) v_{-1} u   v_{ -1}+v   u_1    (1-v   u     ) v_{-1}\end{cases} \label{ut2p}
\\&&\nonumber  \\
&&\begin{cases}
u_{t^-_2}=   (1-u   v     )    (1-u_{-1} v_{-1}   ) u_{-2}+   (u   v
 -1   ) u_{-1} v   u_{ -1}+u   v_1    (u   v  -1   ) u_{-1}\\
v_{t^-_2}=v_2    (u_1 v_1-1   )    (1-u   v     )+v_1    (1-u   v
   ) u_{-1} v  +v_1 u   v_1    (1-u   v     ) 
\end{cases} \label{ut2m}
\end{eqnarray*}
In general, there are relations
$$ u_{t^+}=K,\quad v_{t^+}=-\cT(K)\big|_{u\leftrightarrow v},\qquad
u_{t^-}=\cT(K),\quad
v_{t^-}=- K\big|_{u\leftrightarrow v},
$$
where the reflection $\cT$ is defined by \eqref{taut}.

In the difference algebra $ \fA=\fieldk\langle u_n\, v_n;\, n\in\bbbz\rangle $, we  consider an ideal $\cJ\subset\fA$
generated by an infinite set of polynomials of the form
\begin{equation}\label{cJ}
 \cJ=\langle u_n u_m-\alpha_{n,m} u_m u_n, \ v_n v_m-\beta_{n,m} v_m v_n,\ v_n
u_m-\omega_{n,m} u_m v_n+\eta_{n,m}; \ n,m\in\Z\rangle\, ,
\end{equation}
where $ \omega_{n,m},\alpha_{n,m},\beta_{n,m} \in\C ^*, \eta_{n,m}\in\C,\ n,m\in\Z, $ are parameters.
The ideal $\cJ$ is $\p_{t_0}$--stable. Note that the quotient algebra $\fA_\cJ$ does not always admit a basis
of normally ordered monomials. However, after imposing conditions (i)  and (ii) in Definition \ref{Def1}, we obtain:

\begin{prop}\label{stabt1}
 The ideal $\cJ$ (\ref{cJ}) is $\p_{t_1^+}$--stable and the quotient algebra $\fA_\cJ$ admits a  basis
of normally ordered monomials if and only if
 \begin{equation}
\label{q1}
  \alpha_{n,m}=\beta_{n,m}=1,\quad \omega_{n,m}=\delta_{n,m}(\omega-1)+1,\quad
\eta_{n,m}=\delta_{n,m}(\omega-1),\qquad n,m\in\Z ,
 \end{equation}
where $\omega\in\C ^*$ is an arbitrary non-zero complex parameter and 
$\delta_{n,m}$ denotes the Kronecker symbol. Moreover, the ideal $\cJ$ (\ref{cJ}) satisfying \eqref{q1} is also
$\p_{t_1^-}$--stable.
\end{prop}
Note that the commutation relations (\ref{q1}) are the same as the ones obtained in \cite{Kul81} using  the $r$--matrix technique.
The quantum B\"acklund transformation for such quantum Ablowitz-Ladik chain is studied in \cite{Korff16}.

We denote $\cJ_\omega^{(1)}$   the ideal $\cJ$ (\ref{cJ}) with specialisation
(\ref{q1}).
 In the quotient algebra $\fA_{\cJ_\omega^{(1)}}$ the
commutation relations are given by
\[
 v_n u_n-1=\omega (u_n v_n -1),\qquad  [v_n,v_m] = [v_n,u_m] =  [u_n,u_m]=0,
\qquad n,m\in\Z ,\ n\ne m.
\]
Two quadratic  commuting Hamiltonians
\iffalse
\footnote{I believe that there is also a Hamiltonian $H_0=\frac12
\log (1-uv)+\frac12 \log (1-vu)$ such that
$$u_{t_0}=\frac{1}{\log
(\omega^2)}[H_0,u],\qquad
v_{t_0}=-\frac{1}{\log (\omega^2)}[H_0,v].$$ }
\fi
$$ H_1^+= \sum_{n \in \mathbb{Z}} u_{n+1}v_n,\quad H_1^-= \sum_{n \in
\mathbb{Z}} u_{n-1}v_{n}$$
lead to the Heisenberg representation of $\partial_{t_1^{+}}$ and $\partial_{t_1^{-}}$:
\begin{equation*}
  (\omega^{-1}-1) \, \partial_{t^+_1}(\bullet)=\big[H_1^+ \, , \, \bullet        
\big] \, ; \, \, \, \, \,  (1-\omega)\,  \partial_{t^-_1}(\bullet)=\big[H_1^- \, , \, \bullet \big]\, .
\end{equation*}
The ideal $\fJ^{(1)}_{\omega}$ is a standard quantisation of the non-Abelian Ablowitz-Ladik lattice. 
We construct two quartic Hamiltonians as follows:
\begin{eqnarray*}
&& H_2^+=  \sum\limits_{n\in\Z }\left(
u_{n+1}(2-u_{n }   v_{n }-v_n u_n)   v_{n-1}-u_{n+1}v_nu_{n+1} v_n \right);\label{ALH2}\\
&& H_2^-= \sum\limits_{n\in\Z }\left(
 u_{n-1}(2- u_{n }  v_n-v_n u_n)v_{n+1}-u_{n-1}v_nu_{n-1}v_n\right).\label{ALH2m}
\end{eqnarray*}
In the quotient algebra $\fA_{\cJ_\omega^{(1)}}$, we have
\begin{equation}\label{ALH2d}
  \, \partial_{t^+_2}(\bullet)=\dfrac{1}{  \omega^{-2}-1}\big[ H_2^+  \, , \, \bullet        
\big] \, ; \, \, \, \, \, \,  \partial_{t^-_2}(\bullet)=\dfrac{1}{ 1-\omega^2 }\big[ H_2^-  \, , \, \bullet \big]\, ,
\end{equation}
\iffalse
\begin{equation}\label{ALH2d}
  (\omega^{-2}-1) \, \partial_{t^+_2}=\big[H_2^+ \, , \, \bullet        
\big] \, ; \, \, \, \, \,  (1-\omega^2)\,  \partial_{t^-_2}=\big[H_2^- \, , \, \bullet \big]\, .
\end{equation}
\fi
Next for the family of ideals of the form (\ref{cJ}) we verify conditions (i) and (ii) in Definition \ref{Def1} for the symmetries $\p_{t_2^+}$ and $\p_{t_2^-}$. 
\begin{prop}\label{stabt2}
 There are only four $\p_{t_2^+}$--stable ideals $\cJ_\omega^{(1)},\ldots ,
\cJ_\omega^{(4)}$ of the form $\cJ$ (\ref{cJ}). They are generated by the 
polynomials
   \begin{eqnarray*}
\cJ_\omega^{(1)}  &=&
\label{q21}
\big{\langle}  v_n u_n-1-\omega (u_n v_n -1),\  v_n v_m-v_mv_n,\  u_n u_m-u_m u_n,\
v_nu_m-u_mv_n\,;\,m\ne n \big{\rangle};\\
\cJ_\omega^{(2)}  &=&
\label{q22}
\big{\langle}  v_n u_n-1-\omega (u_n v_n -1),\ v_nu_m-(-1)^{n-m}u_mv_n,\\&&  v_n
v_m-(-1)^{n-m}v_mv_n,\  u_n u_m-(-1)^{n-m}u_m u_n \,;\,m\ne n\big{\rangle};
 \\
\cJ_\omega^{(3)}  &=&
\label{q23}
\big{\langle} v_n u_n-1-(-1)^n\omega (u_n v_n -1),\ v_nu_m-\phi_{n,m}u_mv_n,\
u_nv_m-\phi_{n,m}v_mu_n,\\&&  v_n v_m-\phi_{n,m}v_mv_n,\  u_n u_m-\phi_{n,m}u_m
u_n;\,m>n\big{\rangle};
 \\
\cJ_\omega^{(4)}  &=&
\label{q24}
\big{\langle}  v_n u_n-1+(-1)^n\omega (u_n v_n -1),\ v_nu_m-\phi_{m,n}u_mv_n,\
u_nv_m-\phi_{m,n}v_mu_n,\\&&  v_n v_m-\phi_{m,n}v_mv_n,\  u_n u_m-\phi_{m,n}u_m
u_n;\,m>n\big{\rangle},
 \end{eqnarray*}
where $\phi_{m,n}=\frac{1}{2}(1+(-1)^{n-m}+(-1)^m-(-1)^n)$ and
$\omega\in\C ^*$ is an arbitrary non-zero complex parameter. The ideal $\cJ$ 
(\ref{cJ}) is $\p_{t_2^-}$--stable if and only if it is $\p_{t_2^+}$--stable.
\end{prop}
\begin{proof}
Here, the ideal $\cJ_\omega^{(1)}$ is the same as in Proposition \ref{stabt1}. Its $\p_{t_2^\pm}$--stability follows immediately from the Heisenberg representation for the $\p_{t_2^\pm}$ derivations given by (\ref{ALH2d}). These equations are also valid modulo $\cJ_\omega^{(2)}$.

The ideals $\cJ_\omega^{(3)}$ and $\cJ_\omega^{(4)}$ are neither $\p_{t_1^\pm}$--stable nor $\cS$--stable. They are $\cS^2$--stable and non-standard quantisation ideals of the derivations $\p_{t_2^+}$ and $\p_{t_2^-}$. It follows from $\phi_{n+1,m+1}=\phi_{m,n}$ that $\cJ_\omega^{(3)}=\cS(\cJ_\omega^{(4)})$. Thus the the corresponding quotient algebras are isomorphic $\fA_{\cJ_\omega^{(3)}}\simeq \fA_{\cJ_\omega^{(4)}}$. 

The equations corresponding to $\p_{t_2^+}$ and  $\p_{t_2^-}$ can be presented in the Heisenberg form in $\fA_{\cJ_\omega^{(3)}}$ and $\fA_{\cJ_\omega^{(4)}}$ by 

\begin{equation*}
  \, \partial_{t^+_2}(\bullet)=\dfrac{1}{ 1-\omega^{-2}}\big[(H_1^+)^2 \, , \, \bullet        
\big] \, ; \, \, \, \, \, \,  \partial_{t^-_2}(\bullet)=\dfrac{1}{ \omega^2-1}\big[(H_1^-)^2 \, , \, \bullet \big]\, ,
\end{equation*}
where
\begin{eqnarray*}
  &&  (H_1^+)^2= \sum_{n \in \mathbb{Z}} \left(u_{n+1}v_nu_{n+1}v_n + u_{n+1}(v_nu_n-u_nv_n)v_{n-1}\right);\\
  &&  (H_1^-)^2= \sum_{n \in \mathbb{Z}} \left( u_{n-1} v_n u_{n-1}v_n+u_{n-1} (v_n u_n -u_n v_n) v_{n+1} \right).
\end{eqnarray*}
Thus it follows from Lemma \ref{keylem} that both ideals $\cJ_\omega^{(3)}$ and $\cJ_\omega^{(4)}$ are $\p_{t_2^\pm}$--stable.
\end{proof}

\subsection{The relativistic Toda lattice}\label{RT}
The non-Abelian relativistic Toda system appeared in \cite{Kup00} and is written as
\begin{equation}\label{rToda}
\begin{cases}
u_{t_1}=&u\left(u_{-1}+v\right)-\left(u_1+v_1\right)u\\
v_{t_1}=&vu_{-1}-uv
\end{cases}
\end{equation}
\iffalse
	is Hamiltonian with respect to the Hamiltonian operator
	\begin{equation}
	H_1=\begin{pmatrix}
	0 & \l_u-\r_u\cS\\ \cS^{-1}\l_u-\r_u & \r_u\cS-\cS^{-1}\l_u-\c_v
	\end{pmatrix}
	\end{equation}

	The hierarchy shares the same Lax representation of the commutative case
	\begin{align}
	U&=\begin{pmatrix}
	\lambda v-\lambda^{-1} & u_{-1}\\-1&0
	\end{pmatrix}& B&=\begin{pmatrix}
	-\lambda^{-2}-u_{-1} & \lambda^{-1}u_{-1}\\ -\lambda^{-1} & -u_{-2}-v_{-1}
	\end{pmatrix}
	\end{align}
\fi
Its Hamiltonian and recursion operators are presented in \cite{cw19-2}, from which we obtain its
next symmetry flow:
\begin{equation} \label{secondflowrt}
	\begin{cases}
	u_{t_2}=&u\left(u v +v^2+u_{-1} v +v u_{-1} +u u_{-1} +u_{-1}^2+u_{-1} v_{-1} +u_{-1} u_{-2} \right), \\
	& - \left( v_1 u+u_1 u +u_1^2+ u_1 v_1 +v_1 u_1 +v_1^2+u_2 u_1 +v_2 u_1\right)u\\
	v_{t_2}=&v(u_{-1}u_{-2}+u_{-1}^2+u_{-1} v_{-1}+u u_{-1}+v u_{-1})-(u^2+u v +u u_{-1}+ u_1 u+v_1 u) v.
	\end{cases}
	\end{equation}
    The relativistic Toda lattice admits a standard bi-quantisation ideal $\cJ_{\omega,\eta}$ generated by
	\[
	 \begin{array}{lll}
	 u_nu_{n+1}-\omega u_{n+1}u_n,\quad & v_n v_{n+1}- v_{n+1}v_n-\eta u_n  \,\,  ,\quad & v_nu_{n}-\omega u_{n}v_n+\eta u_n  \,,\\
	v_nu_{n+1}-u_{n+1}v_n,&u_nv_{n+1}-\omega v_{n+1}u_n+\eta u_n,&\\
	u_nu_{m}-u_{m}u_n,& 	v_n v_{m}-v_{m}v_n,& 	v_nu_{m}-u_{m}v_n,\quad |n-m|>1.
	 \end{array}
	\]
The first two commuting Hamiltonians of the quantum hierarchy are given by 
\begin{eqnarray}\label{H1rt}
&& H_1=\sum\limits_{n\in\Z } (u_n+v_n)\\ \label{H2rt}
&& H_2=\frac12 \sum\limits_{n\in\Z }\Big( (u_n+v_n)^2+u_n (u_{n+1}+v_{n+1})+( u_{n+1}+v_{n+1}) u_n+\eta u_n\Big)
\end{eqnarray}
These should be part of an infinite family of commuting Hamiltonians $\{H_{k \in \mathbb{N}}\}$. The Heisenberg form of the quantum hierarchy can be found solving the triangular system 
\begin{equation} \label{Hampen2}
    \begin{split}
        [H_1, \, \bullet \, ]&=(\omega -1) \partial_{t_1}(\bullet) , \\
        [H_2, \, \bullet \, ]&= \frac{\omega^2-1}{2} \partial_{t_2}(\bullet)+\eta\, \omega \partial_{t_1}(\bullet).
%[H_3, \, \bullet \, ] &= \frac{\omega^3-1}{3}\partial_{t_3}+ ...
        \end{split}
\end{equation}

The second flow in the relativistic Toda lattice \eqref{secondflowrt} admits a non-standard quantisation ideal $\cI_{\omega}$ generated by
	\[
	 \begin{array}{lll}
	 u_nu_{n+1}-(-1)^n\omega u_{n+1}u_n,\quad & v_n v_{n+1}+ v_{n+1}v_n,\quad & v_nu_{n}+(-1)^{n}\omega u_{n}v_n,\\
	v_nu_{n+1}+u_{n+1}v_n,&u_nv_{n+1}-(-1)^n\omega v_{n+1}u_n,&\\
	u_nu_{m}+u_{m}u_n,& 	v_n v_{m}+v_{m}v_n,& 	v_nu_{m}+u_{m}v_n,\quad |n-m|>1.
	 \end{array}
	\]
The Heisenberg derivation in $\fA_{\cI_{\om}}$ takes the form 
$$ \partial_{t_2}(\, \bullet \,) = \frac{2}{\omega^2-1}[ \, \hat{H}_2 \, , \, \bullet \, ] .$$

Here the Hamiltonian $\hat{H}_2$ is equal to $H_1^2$, where $H_1$ is given by (\ref{H1rt}). In the quantum algebra $\fA_{\cI_{\om}}$, it coincides with $H_2\Big|_{\eta=0}$ defined by (\ref{H2rt}).

\subsection{The Merola-Ragnisco-Tu lattice}
The classical Merola-Ragnisco-Tu lattice was introduced in \cite{MRT}, and its non-Abelian version in \cite{cw19-2}:
\begin{eqnarray}\label{mrteq}
		\begin{cases} u_{t_1}= u_{1}-u v u\\ v_{t_1}=  -v_{-1}+v uv \end{cases}
\end{eqnarray}
The symmetry hierarchy starts with
$u_{t_0}=u,  \ v_{t_0}= -v
$
and the next symmetry is
	\begin{eqnarray}\label{sym2MRT}
		\begin{cases} u_{t_2}= u_{2}-(u_1 v_1+uv) u_1- (u_1 v+u v_{-1}) u +uv u v u\\
		v_{t_2}=  -v_{-2}+v_{-1} (u_{-1} v_{-1}+uv)+v (u v_{-1}+u_1 v) -vuv uv \end{cases}
	\end{eqnarray}
Both (\ref{mrteq}) and (\ref{sym2MRT}) admit a standard biquantisation ideal $\cJ_{\omega,\eta}$  generated by the relations
	\[
	 \begin{array}{llll}
	  v_{n+1}u_n+1-\omega (u_n v_{n+1}+1),\quad & u_nv_n-\omega v_n u_n+\eta,\ &v_{m+1} u_n-\omega u_n v_{m+1},\\
	  u_nu_m-\omega u_m u_n,& v_nv_m-\omega v_m v_n,& u_mv_n-\omega v_n u_m,
	 \end{array}
	\]
where $n, m \in \mathbb{Z}$ and $m > n$. In quantum algebra $\fA_{\cJ_{\omega,\eta}}$, the following two quantum Hamiltonians
\begin{eqnarray*}
&&H_1=\sum_{n\in\Z} v_n u_n, \\
&&H_2=\frac12 \sum_{n\in\Z} (v_n u_{n+1}+u_{n+1}v_n-v_n u_n^2 v_n)
\end{eqnarray*}
commute. These are part of an infinite family of commuting Hamiltonians $\{H_{n \in \mathbb{N}}\}$. The Heisenberg form of the quantum hierarchy can be found by solving the triangular system 
\begin{equation} \label{Hampen3}
    \begin{split}
        [H_1, \, \bullet \, ]&=(\omega -1) \partial_{t_1}(\bullet) , \\
        [H_2, \, \bullet \, ]&= \frac{\omega^2-1}{2} \partial_{t_2}(\bullet)+\eta \,\omega \partial_{t_1}(\bullet).
        \end{split}
\end{equation}
    
%\subsubsection{Alternative quantisation}\label{secMRT2}
Similarly to the Volterra and Toda lattices we discussed in Section \ref{sec32}, the symmetry (\ref{sym2MRT}) also admits a non-standard quantisation ideal $\cJ_b$ generated by the relations
\[
 \begin{array}{ll}
u_p u_q-(-1)^{p-q}\omega u_q u_p,\ \quad p <q, & v_p v_q-(-1)^{p-q}\omega v_q v_p, \quad p <q, \\ 
\omega u_p v_q+(-1)^{p-q} v_q u_p,\quad  p<q-1,\ & u_p v_q+(-1)^{p-q}\omega v_q u_p,\quad p \geq  q, \\  \omega (u_p v_{p+1}+1)-(v_{p+1}u_p +1).  &   
 \end{array}
\]
Let
\begin{eqnarray*}
\hat{H}_2=\frac12 \sum_{n\in\Z } (v_n u_{n+1}+u_{n+1}v_n+v_n u_n^2 v_n).
\end{eqnarray*}
Then in the quantum algebra $\fA_{\cJ_b}$, we have $[H_1,\hat{H}_2]=0$ and
$$
 \partial_{t_2}(\bullet)=\frac{2}{\omega-\omega^{-1}}[\hat{H}_2,\, \bullet\, ]  .
$$

\subsection{The Adler-Yamilov lattice} \label{ssec:Adler-Yamilov}
In the commutative case, the Adler-Yamilov lattice\footnote{Adler and Yamilov derived equation (\ref{AY}) as a Darboux chain for Kaup’s integrable PDE. Although often called the Kaup lattice \cite{cw19-2,KMW}, it is more accurate to attribute it to Adler and Yamilov.} appeared in \cite{AdYam94}. Its non-Abelian version was introduced and studied in \cite{cw19-2}
\begin{equation}\label{AY}
		\begin{cases} u_{t_1}=(u_1-u) (u+v) \\ v_{t_1}= (u+v) (v-v_{-1}) \end{cases}
\end{equation}
Using the recursion operator from \cite{cw19-2}, one can obtain the next member of its symmetry hierarchy:
\[ 
		\begin{cases} u_{t_2}=\left((u_2-u_1) (u_1+v_1)+(u_1-u) 
(u_1+v)\right) (u+v)-(u_1-u) (u+v)(v-v_{-1}) \\
		v_{t_2}= (u_1-u) (u+v)(v-v_{-1})+(u+v) \left((u+v_{-1})(v-v_{-1})+ (u_{-1}+v_{-1})(v_{-1}-v_{-2})\right) \end{cases}
	\]
The non-Abelian Adler-Yamilov lattice admits a quantisation ideal $\cI_{\beta}$ generated by the commutation relations
\begin{equation*}
    u_nv_n-v_nu_n- \beta (u_n+v_n),\  u_n v_m - v_m u_n,\  u_n u_m -u_m u_n,\  v_n v_m - v_m v_n,\quad  n\in\bbbz,\ n\ne m \, .
\end{equation*}
This yields a standard quantisation of the Adler–Yamilov equation, since the quotient algebra becomes commutative when $\beta = 0$. In the quantum algebra $\fA_{\cI_\beta}$, the following two local functionals commute:
 \begin{eqnarray*}
H_1&=&\sum_{n\in\Z} \left(u_nv_n-v_nu_{n+1}\right), \\
H_2&=& \sum_{n\in\Z} \left( u_{n+1}v_{n+1}v_n+u_{n+1}u_nv_n-u_{n+2}u_{n+1}v_n-u_{n+2}v_{n+1}v_n \right) .
\end{eqnarray*}
 These should be part of the infinite family of commuting Hamiltonians $\{ H_{ n \in \mathbb{N}} \}$. When $\beta \neq 0$, these first two members of the quantum Adler-Yamilov hierarchy can be presented in the Heisenberg form:
\[
\left\{\begin{array}{l}
\partial_{t_1}u_n=\frac{1}{\beta}\big{[}H_1 \, , \, u_n \, \big{]} ,  \\
\partial_{t_1}v_n=\frac{1}{\beta}\big{[}H_1 \, , \, v_n \, \big{]},
\end{array}\right.
 \quad 
\left\{\begin{array}{l}
   \partial_{t_2}u_n=\frac{1}{\beta}\big{[}H_2 -\beta H_1\, , \, u_n \, \big{]} ,  \\
    \partial_{t_2}v_n=\frac{1}{\beta}\big{[}H_2-\beta H_1 \, , \, v_n \, \big{]}.
\end{array}\right.
\]

\subsection{The Chen-Lee-Liu lattice}\label{ssec:cll}
In the classical commutative case, the differential-difference equation
\begin{equation}\label{utcll0}
\begin{cases}
u_{t_1}=(u_1 - u)(1+vu),\\
v_{t_1}=(1+vu)(v - v_{-1});
\end{cases}
\end{equation}
was obtained by Tsuchida \cite{Tsuchida} as an integrable discretisation of the derivative nonlinear Schr\"odinger type equation discovered by Chen, Lee and Liu and was further studied in \cite{KMW}. Its integrable non-Abelian analogue with variables belonging to the free algebra $\fA= \mathbb{C} \langle \, u_n, v_n \, | \,  n \in \mathbb{Z} \, \rangle$ was introduced in \cite{cw19-2}. Using the recursion operator from \cite{cw19-2}, we obtain the next member of the symmetry hierarchy:
\begin{equation}\label{utcll1}
		\begin{cases} u_{t_2}=(u_2-u_1) (1+v_1 u_1) (1+vu) +(u_1-u) v (u_1-u) (1+vu)+(u_1-u) (1+vu) (v_{-1}u-1) \\
		v_{t_2}= (1+vu) (1+v_{-1}u_{-1})(v_{-1}-v_{-2})+v (u_1-u)(1+vu)(v-v_{-1})+(1+vu) (v_{-1}u-1) (v-v_{-1}) \end{cases}
\end{equation}
A quantisation ideal $\cI_\omega\subset\fA$ is generated by polynomials 
\[
 u_n v_n -\omega v_n u_n+1-\omega,\  u_n v_m - v_m u_n,\  u_n u_m -u_m u_n,\  v_n v_m - v_m v_n,\quad  n\in\bbbz,\ n\ne m\, , 
\]
where $\omega\ne 0$ is a quantisation parameter. This quantisation is standard since at $\omega=1$ the quotient algebra $\fA/\cI_1$ becomes commutative.
The first two commuting quantum Hamiltonians are given by:
\begin{eqnarray*}
 &&H_1=\sum_{n\in\Z } (u_{n+1}-u_n)v_n;\\
 &&H_2=\frac{1}{2}\sum_{n\in\Z }\Big( 
 u_{n}v_nu_{n}v_{n-1}
 +u_{n+1}v_nu_nv_{n}+u_{n}v_{n-1}u_{n}v_{n}+ u_nv_nu_{n+1}v_n-u_{n+1}v_nu_{n+2}v_{n+1}
 \\&& \, \, \, \qquad-
 u_nv_nu_nv_n-u_{n+1}v_nu_{n+1}v_n-u_{n+1}v_nu_nv_{n-1} 
 +u_{n+2}(4v_{n+1}-2v_n-2 v_{n+2}) \Big).
\end{eqnarray*}
The quantum equation (\ref{utcll0}) and its first symmetry (\ref{utcll1}) in the quantum algebra $\fA/\cI_\omega$ can be expressed in Heisenberg form in terms of these two Hamiltonians:
\[
\left\{\begin{array}{l}
\partial_{t_1}u_n=\dfrac{1}{1-\omega}\big{[}H_1 \, , \, u_n \, \big{]} ,  \\
\partial_{t_1}v_n=\dfrac{1}{1-\omega}\big{[}H_1 \, , \, v_n \, \big{]},
\end{array}\right.
 \quad 
\left\{\begin{array}{l}
   \partial_{t_2}u_n=\dfrac{2}{1-\omega^2}\big{[}H_2 \, , \, u_n \, \big{]} ,  \\
    \partial_{t_2}v_n=\dfrac{2}{1-\omega^2}\big{[}H_2 \, , \, v_n \, \big{]}.
\end{array}\right.
\]

\subsection{The Belov-Chaltikian lattice}
The Belov-Chaltikian lattice is the Boussinesq lattice related to the lattice $W_3$-algebra \cite{BC93}. Its
 non-Abelian version appeared in \cite{cw19-2}. It is defined on the non-Abelian algebra of difference polynomials $\fA= \mathbb{C} \langle \, u_n, v_m \, | \,  m , n \in \mathbb{Z} \, \rangle$ by
\[
\begin{cases}
u_{t_1}=v_{2}u-u v_{-1},\\
v_{t_1}=v_{1}v-v v_{-1}+u_{-1}-u .
\end{cases}
\]
Its next symmetry is
\[
\begin{cases}
u_{t_2}=(v_3 v_2+v_2^2+v_2 v_1) u -u (v v_{-1}+v_{-1}^2+v_{-1} v_{-2})-(u_2+u_1)u+u(u_{-1}+u_{-2}),\\
v_{t_2}=(v_2 +v_1)(v_1 v-u)+v_1 v^2-v^2 v_{-1}-(v v_{-1}-u_{-1}) (v_{-1}+v_{-2})-(u_1+u) v+v (u_{-1}+u_{-2}).
\end{cases}
\]
Let $\cI_{\omega}$ be the ideal generated by 
\begin{equation}
    \begin{split}
        u_nv_{n+2}-\omega v_{n+2}u_n & , \, \, \, \, v_{n-1}u_n -\omega u_nv_{n-1},  \, \, \, \, v_nv_{n+1}-\omega v_{n+1}v_n + (\omega-1) u_n, \\
u_nu_{n+1}-\omega u_{n+1}u_n &, \, \, \, \, u_nu_{n+2}-\omega u_{n+2}u_n
    \end{split}
\end{equation}
and the commutators of all remaining pairs of generators. This ideal is stabilised by the derivation $\partial_{t_1}$ and is a deformation of the first Poisson bracket of the classical Belov-Chaltikian lattice \cite{BC93, KMW}. Moreover,  in the quotient algebra $\fA_{\cI_{\omega}}$ the Belov-Chaltikian equation and its next symmetry can be put in Heisenberg form:
\begin{eqnarray*}
&& \partial_{t_1}(\, \bullet \,)=\frac1{1-\omega}\big{[} \, H_1\, ,\, \bullet \, \big{]},\qquad  H_1=\sum_{n \in \mathbb{Z}} v_n;\\
&&\partial_{t_2}(\, \bullet \,)= \frac{2}{1-\omega^2}\big{[} \, H_2 \, ,\, \bullet \, \big{]}, \qquad
H_2=  \frac{1}{2}\sum_{n \in \mathbb{Z}} v_n^2+v_{n+1}v_n+v_nv_{n+1} -2u_n .
\end{eqnarray*}
\iffalse
Hence for any $a \in \fA$ we have
$$a_{t_1} = \frac{1}{\omega}\Big{[}\, \sum_{n \in \mathbb{Z}} v_n \, ,\, a \, \Big{]} \, \, \text{ \, mod \, } \cI_{\omega} $$
which means that $\cI_{\omega}$ is stabilized by the derivation $\partial_{t_1}$.
\fi

\subsection{The Blaszak-Marciniak lattice}
The classical commutative version of the Blaszak–Marciniak lattice first appeared in \cite{BlaszMarc} and was further studied in \cite{KMW}. In \cite{cw19-2}, an integrable lift to the free algebra of difference polynomials 
$$\fA= \mathbb{C} \langle \, u_n, v_m, w_k \, | \, k, m , n \in \mathbb{Z} \, \rangle$$
was presented:
\[
\begin{cases}
u_{t_1}=w_1-w_{-1},\\
v_{t_1}=w_{-1} u_{-1}-u w,\\
w_{t_1}=wv-v_1 w;
\end{cases}
\]
 Using the recursion operator from \cite{cw19-2}, we obtain the next member of its symmetry hierarchy:
\[
\begin{cases}
u_{t_2}=v_2 w_1+ w_1 v_1 - v w_{-1}- w_{-1} v_{-1},\\
v_{t_2}=v w_{-1} u_{-1}-u w v +w_{-1} v_{-1} u_{-1} - uv_1 w +w_1 w -w_{-1} w_{-2},\\
w_{t_2}=wv^2 +w w_{-1} u_{-1}-u_1 w_1 w-v_1^2 w.
\end{cases}
\]
A $\p_{t_1}$--stable difference ideal $\fJ_{\eta,\Omega}$ is generated by the  following set of polynomials
\[
 \begin{array}{ll}
  \Omega uu_1- u_1 u+(\Omega -1)v_1+\eta,\quad &uu_n-\Omega^{(-1)^{n}}u_n u \ \  (n>1),\\
  vu_1-u_1v-(1-\Omega)w, & v_2u-uv_2+(1-\Omega)w_1\\
  wu_n-\Omega^{(-1)^{n+1}}u_nw\ \  (n>0),&uw_n-\Omega^{(-1)^{n+1}}w_nu\ \  (n>0),\\
  vv_1-v_1 v+(1-\Omega)uw,&ww_n-\Omega^{(-1)^{n}}w_n w\ \  (n>1),\\
  \Omega wv-vw+\eta w,& wv_1-\Omega v_1w-\eta w,\\
  vu_n-u_nv \ \ (n\notin \{1,-2\}), &
  wv_n-v_nw\ \  (n\notin \{ 0,1\}),\\
  vv_n-v_n v\ \ (n\ne 1), &
  wu-uw  , \\   ww_1-w_1w,&
 \end{array}
\]
and their shifts by $\cS^m,\ m\in\Z$, where $\Omega\ne 0$ and $\eta$ are  quantisation parameters. We use $\Omega$ instead of $\omega$ to avoid confusion with $w$.

The first three commuting Hamiltonians in the quantum algebra $\fA_{\eta,\Omega}=\fA/\fJ_{\eta,\Omega}$ are:
\begin{eqnarray*}
 &&H_1=\sum_{n\in\Z}v_n,\\
 &&H_2=\frac12\sum_{n\in\Z} ( v_n^2+2u_nw_n+u_nu_{n+1}v_n-u_nv_nu_{n+1}),
 \\
&&H_3= \frac13\sum_{n\in\Z} \big( v_n^3+2u_nw_nv_n+u_nw_nv_n+2u_nv_{n+1}w_n +u_nw_nv_{n+1} +u_n(v_nu_{n+1}-u_{n+1}v_n)v_{n+1} 
\\&&\qquad-3w_{n+1}w_n+(v_{n+2}u_n-u_nv_{n+2}+u_nv_{n+1}+u_nv_n+3w_{n+1})(v_nu_{n+1}-u_{n+1}v_n) \big)
\end{eqnarray*}
These should be part of an infinite family of commuting Hamiltonians $\{H_{n \in \mathbb{N}}\}$. The Heisenberg form of the quantum hierarchy can be found solving the triangular system 
\begin{equation} \label{Hampen4}
    \begin{split}
        [H_1, \, \bullet \, ]&=(\Omega -1) \partial_{t_1}(\bullet) , \\
        [H_2, \, \bullet \, ]&= \frac{\Omega^{2}-1}{2} \partial_{t_2}(\bullet)+\eta \Omega \partial_{t_1}(\bullet), \\
        [H_3, \, \bullet \, ] &= \frac{\Omega^3-1}{3}\partial_{t_3}(\bullet)+ \eta \Omega^2 \partial_{t_2}(\bullet)+ \eta^2 \Omega \partial_{t_1}(\bullet).
        \end{split}
\end{equation}

\iffalse
Taking the specialisations either $\omega=1$ or $\eta=0$, we define two quantum algebras $\fA_\eta=\fA/\fJ_{\eta,1}$ with $\eta\neq 0$, and $\fA_\omega=\fA/\fJ_{0,\omega}$ with $\omega^l\neq 1, l\in \N$.
On these algebras the quantum Blaszak-Marciniak hierarchy  can be presented in the Heisenberg form:
\[
\begin{array}{ll}
 \mbox{for any }\ a\in \fA_\eta\qquad &\mbox{for any }\ b\in\fA_\omega\\ &\\
\phantom{ww} [H_1,a ]=0,&\p_{t_1}(b)=\dfrac{\omega}{1-\omega}[H_1,b],\\ &\\
\p_{t_1}(a)=\dfrac{1}{ \eta}[H_2,a],&\p_{t_2}(b)=\dfrac{2\omega^2}{1-\omega^2}[H_2,b],\\ &\\
\p_{t_2}(a)=\dfrac{1}{ \eta}[H_3,a],&\p_{t_3}(b)=\dfrac{3\omega^3}{1-\omega^3}[H_3,b].
\end{array}
\]

Representing $\Omega=1+\alpha \mu$ and $\eta=\beta \mu$ and taking the classical limit $\mu\to 0$ of the commutator on $\fA_{\eta,\Omega}$, we recover the well-known Poisson pencil.
\[
 \a \{a,b\}_1+\b\{a,b\}_2=\lim\limits_{\mu\to 0}\frac{1}{\mu}[a,b]
\]
associated with the Blaszak-Marciniak hierarchy \cite{BlaszMarc, KMW}.

\fi

\section{Summary and Discussions}
In this paper, we develop the recently proposed algebraic quantisation approach based on non-Abelian lifts and quantisation ideals \cite{AvM20}.  Applying this new method to several non-Abelian integrable differential-difference systems, we reveal rich quantisation structures and patterns.
We obtained not only known standard quantisation but also new bi-quantisation ideals and non-standard quantisation ideals. 
A quantisation ideal is non-standard if no values of the quantum parameters of the quantum algebra yield a commutative algebra; conversely, for a standard quantisation, the quantum algebra admits
a commutative limit\footnote{For the non-standard quantisation, we previously used the term ``non-deformation quantisation'' in \cite{CMW, CMW2}. However, such a quantisation could also be understood as a deformation of a noncommutative algebra \cite{MV}.}.

For the non-Abelian Toda hierarchy, we demonstrate that each symmetry within its hierarchy possesses a standard bi-quantum structure. This structure, determined by a quantisation ideal with two arbitrary parameters, corresponds to the Poisson pencil of the classical Toda hierarchy. Furthermore, the even Toda sub-hierarchy admits a non-standard quantisation. These results are established through the Miura transformation linking the Toda and Volterra hierarchies. By applying this transformation, we also derive the Hamiltonians of the Toda hierarchy from those of the Volterra hierarchy given in \cite{CMW2}.

For the non-Abelian Bogoyavlensky lattice hierarchy \eqref{nbogadd}, we show that for any $p \in \mathbb{N}$, each member admits a standard quantisation with commutation relations
\begin{eqnarray*}\label{bogp}
u_n u_{n+k} = \omega , u_{n+k} u_n , \quad 1 \leq k \leq p, \qquad
u_n u_m = u_m u_n , \quad |n-m| > p, \quad n,m \in \mathbb{Z},
\end{eqnarray*}
where $\omega$ is a nonzero quantisation constant. The proof proceeds similar to the case of the Volterra lattice \cite{CMW}, relying on explicit formulae for the hierarchy.

In Section \ref{Sec4}, we present quantisation results for several lattice equations. For the modified Volterra, Ablowitz–Ladik, relativistic Toda, and Merola–Ragnisco–Tu lattices, we derive both standard quantisation and non-standard quantisation for their next symmetry. Additionally, we obtain a standard bi-quantum structure for the relativistic Toda, Merola–Ragnisco–Tu, and Blaszak–Marciniak lattices. 
\iffalse 
Finally, we note that the similarity among equations \eqref{Hampen1}, \eqref{Hampen2}, \eqref{Hampen3}, and \eqref{Hampen4} suggests that Hamiltonians independent of the quantisation parameters can always be constructed. 
\fi

We introduce the concept of local functionals with respect to quantisation ideals in Section \ref{sec23}. The Hamiltonians of the quantum Heisenberg equations must be of this type. In the present paper, we list the first few commuting Hamiltonians without providing their derivation. They are obtained from known Lax representations of the corresponding non-Abelian equations, and the details of this derivation will be presented in a forthcoming article.

The relationship between quantisation ideals and Miura transformations requires further study. 
As discussed in Section \ref{sec41}, the symmetry of the modified Volterra lattice \eqref{m2V1} relates to that of the Volterra lattice \eqref{vol} via the Miura transformation $u_n = v_{n+1} v_n$. Equation \eqref{m2V1} admits three quadratic non-standard quantisation ideals, while \eqref{vol} admits only one.  The relation between their quantisation ideals is given by \eqref{VmV}.

Miura transformations also affect the number of parameters in quantisation ideals. Consider the modified non-Abelian Bogoyavlensky chain
\begin{equation}\label{mb2}
w_t=w_{2} w_{1}  w^2 + w_{1}  w w_{-1}  w - w  w_{1}  w  w_{-1} - w^2 w_{-1}  w_{-2}\, .
\end{equation}
Under the Miura transformation $u=w_2 w_1 w$, this system is transformed into the Bogoyavlensky chain (\ref{nbogadd}) when $p=2$:
\begin{equation}\label{bogp2}
u_t=(u_2+u_1)u-u (u_{-1}+u_{-2}).
\end{equation}
System (\ref{mb2}) admits a quadratic quantisation ideal generated by, for $n>m$,
\[
w_m w_n=\left\{\begin{array}{ll}\alpha w_n w_m, & n-m=1 \mod 3;\\
\beta w_n w_m, & n-m=2 \mod 3;\\ \alpha^{-1} \beta^{-1} w_n w_m, & n-m=0 \mod 3,
\end{array}\right.
\]
where $\alpha\ne 0,\ \beta\ne 0$ \cite{AvM20}.
Notably, the quantisation of (\ref{mb2}) depends on two independent parameters, $\alpha$ and $\beta$, whereas the quantisation of (\ref{bogp2}) involves only a single parameter $\omega = \alpha \beta$. For $p = 3$, the quantisation of the corresponding modified chain depends on four independent parameters. As $p$ increases, the number of independent parameters in the quantisation ideals also grows, while for the Bogoyavlensky chain (\ref{nbogadd}) it remains a single one.

In Section \ref{sec41}, we showed that the quantisation ideals of two distinct non-Abelian integrable lifts of the modified Volterra equation coincide. The standard quantisation gives the same quantum equation for both lifts, while non-standard ones yield two different quantum equations with distinct Hamiltonians. In general, however, the problem of how different non-Abelian integrable lifts influence the resulting quantum hierarchies has not yet been addressed.

Bi-Hamiltonian structures, recursion operators, and Lenard's scheme are powerful tools widely used in the theory of integrable systems. The algebraic quantisation approach provides a framework for constructing quantisation ideals that depend on several parameters, and it is natural to refer to such systems and algebras as multi-quantum.
The quantum analogues of the recursion operators and Lenard's scheme are yet unknown. 

\section*{Acknowledgements}
SC is supported by BK21 SNU Mathematical Sciences Division.
AVM gratefully acknowledges the hospitality and support received from the
School of Mathematics and Statistics, Ningbo University. 
JPW is supported by the Ningbo University Research Start-Up Fund No. 029-432514993.
This article is partially based upon work from COST Action CaLISTA
CA21109 supported by COST (European Cooperation in Science and Technology).

\section*{Appendix: Proof for Theorem \ref{BOGth}}

Similarly as for the Volterra chain, describing the monomials appeared in $Y^{(l)}$ defined by \eqref{yl}, we introduce two sets for $\a=(\a_1, \a_2, \cdots, \a_k)\in \bbbz^k$ as follows:
 \begin{eqnarray*}
  && \cA^k=\left\{\a\in \bbbz^k\big{|} (1-k)p\leq \a_{k} \leq p-1,\ 0\leq\a_1 \leq kp-1,\ \a_{i+1}+p \geq \a_i,\ i=1,...,k-1\right\};\\
  &&\cZ^k_{\geq}=\left\{\a\in \bbbz^k\big{|} \a_{i+1}+p \geq \a_i\geq a_{i+1},\ i=1,...,k-1\right\}.
 \end{eqnarray*}
We now simply write the expression $Y^{(k)}$ given by (\ref{yl}) as $Y^{(k)}=\sum_{\a\in \cA^k} u_{\a}$
Given an ideal $\fI$ defined by \eqref{nibq}, the canonical projection $\pi: \fA \rightarrow \fA/\fI $ acts $Y^{(k)}$ as follows:
\begin{eqnarray*}
 \pi (Y^{(k)})=\sum_{\a\in \cA^k\cap \cZ^k_{\geq}} P_{\a}(\omega) u_{\a},
\end{eqnarray*}
where $P_{\a}(\omega)$ the unique polynomial in $\mathbb{Z}_+[\omega]$ such that for $\a\in  \cA^k\cap \cZ^k_{\geq}$,
 \begin{equation}\label{paB}
 P_{\a}(\omega)u_{\a} =\pi \left( \sum_{\beta\in \cA^k, \beta \sim \a} u_{\beta}\right).
 \end{equation}
%It is clear that if $\alpha\in \cA^k\cap \cZ^k_{\geq}$, then $\#_{\a}(0)\geq 1$.
Clearly, we have $\pi(Y^{(1)})=Y^{(1)}$. Note that $Y^{(2)}=\sum_{i=0}^{2p-1}\sum_{j=0}^i u_j u_{i-p}$. Thus
\begin{eqnarray*}
 \pi(Y^{(2)}\big{|}_{p=2})=u_3 u_1+u_2 u_1+u_1^2+u_2 u+(1+\om) u_1 u+u^2 +u_1 u_{-1}+u u_{-1}+u u_{-2},
\end{eqnarray*}
which implies $P_{(1,0)}(\om)\big{|}_{p=2}=1+\om$.

Here we give the explicit formula for the function $P_{\a}(\om)$,
which is valid for $p=1$ corresponding to the Volterra lattice \eqref{vol} in \cite{CMW2}.
\begin{Pro}\label{pro1B} Let $\a\in \cZ^k_{\geq}$.
The expression of $P_{\a}(\om)$ can be written in terms of quantum binomial
coefficients as follows:
\begin{equation}\label{Pf}
P_{\a}(\om)= \prod_{r \geq p} \binom{ \sum_{l=0}^{p} \#_{\a}(r-l) -1}{\#_{\a}(r)}_{\om} \, . \,\prod_{r=0}^{p-1} \binom{\sum_{l=0}^{r} \#_{\a}(l)}{\#_{\a}(r)}_{\om}  \, . \,
 \prod_{r <0} \binom{\sum_{l=0}^{p} \#_{\a}(r+l) -1}{\#_{\a}(r)}_{\om},
\end{equation}
where $p$ is associated to the equation \eqref{nbogadd}, $\#_{\a}(r)$ is the number of $r$ in $\a$ and
$$
\binom{n}{r}_{\om}= \frac{(\om^n-1)(\om^{n-1}-1)\cdots (\om^{n-r+1}-1)}{(\om^r-1)(\om^{r-1}-1)\cdots (\om-1)}.
$$
\end{Pro}
Before we give a proof for this statement, we first prove the following lemma:
\begin{Lem}\label{lem1} Given $x\in\fA$ and $y\in\fA$, assume that $x y= \om \, yx $, and let $z=z^{(1)}...z^{(n)}$.  For all $n,k \geq 0$, we define
\begin{eqnarray*}
&& A_{k,n} = \{z |z^{(i)} =x \, \,  \text{or} \, \, z^{(i)} = y, \, \mbox{for $1\leq i\leq n$ and the number of $x$ is $k$}\}; \\
&& B_{k,n} =   \{z \in A_{k,n} | z^{(1)} =  y\} \quad \mbox{and} \quad C_{k,n} =   \{z \in A_{k,n} | z^{(n)} =  x\}.
\end{eqnarray*}
Then $\sum_{z \in A_{k,n}} z = \binom{n}{k}_{\om} y^{n-k}x^k$ and $\sum_{z \in B_{k,n}} z = \sum_{z \in C_{k,n}} z= \binom{n-1}{k}_{\om} y^{n-k}x^k$.
\end{Lem}
\begin{proof}
Once we prove the result for $A_{k,n}$, the second part follows immediately. We do it by induction on $n$. If $n=1$ the result is true. It is also true for any $n$ when $k = 0$. Suppose it holds up to the integer $n$ for any $k$. Then for any positive integer $k$ we can split the set $A_{k,n+1}$ as follows:
\begin{equation*}
 A_{k,n+1} = \{z \in A_{k,n+1} \, | z_{n+1} = x \} \oplus \{z \in A_{k,n+1} \, | z_{n+1} = y \}
\end{equation*}
By induction, we have
\begin{eqnarray*}
&&\sum_{z \in A_{k,n+1}} z = \binom{n}{k-1}_{\om} y^{n+1-k} x^k + \om^k  \binom{n}{k}_{\om} y^{n+1-k} x^k\\
&&\qquad =\left(\binom{n}{k-1}_{\om} + \om^k  \binom{n}{k}_{\om} \right)y^{n+1-k} x^k=\binom{n+1}{k}_{\om} y^{n+1-k}x^k,
\end{eqnarray*}
which follows from $\om^k-1 + \om^k(\om^{n-k+1} -1) = \om^{n+1}-1$. Thus we proved the required formulas.
\end{proof}

{\bf Proof of Proposition \ref{pro1B}.}
 We first prove the result when $\alpha\in \cZ^k_{\geq}$ has no negative components by induction. Let $m$ be the maximum integer in $\alpha$. Let $\beta$ be obtained from $\alpha$ after deleting all
 the instances of $m$. Suppose first that $1\leq m < p$. A monomial equivalent  to $\alpha$ is given by a monomial equivalent to $\beta$ and an element in the set
 $A_{\#_{\a}(m),k}$, where $x=u_m$, $y=u_i, 0\leq i<m $ and $u_i u_m = \om \, u_m u_i$ as in (\ref{nibq}). In this case, $k=\sum_{l=0}^{m} \#_{\a}(l)$. Therefore it follows from the Lemma \ref{lem1} that
 \begin{equation*}
 P_{\alpha}(\om) = \binom{\sum_{l=0}^{m} \#_{\a}(l)}{\#_{\a}(m)}_{\om} P_{\beta}(\om),
 \end{equation*}
which is also valid for $m=0$ since in this case $P_{\a}(\om)=1$.

 If we have $m \geq p$, we cannot have an instance of $m$ at the very end of a monomial equivalent to $\alpha$ according to the definition of $\cA^k$. Hence a monomial equivalent to $\alpha$ is given by a monomial equivalent
 to $\beta$ and an element of the set $A_{\#_{\a}(m),n}$, where $x=u_m$, $y=u_i, m-p\leq i<m $,  $n =  \sum_{l=0}^{p} \#_{\a}(m-l)$ and  $z_n = y$.  Therefore following from
 Lemma \ref{lem1}, we have
  \begin{equation*}
 P_{\alpha}(\om) = \binom{ \sum_{l=0}^{p} \#_{\a}(m-l) -1}{\#_{\a}(m)}_{\om} P_{\beta}(\om).
 \end{equation*}
 We now prove the general case when $\a$ has negative components by induction on the minimal integer in $\a$ that we denote by $q$. Let $\beta$ be obtained from $\a$ after deleting all the instances of $q$.
 A monomial equivalent to $\a$ is given by a monomial equivalent to $\beta$ together with an element of the set $A_{\#_{\a}(q),n}$, where $x=u_q$, $y=u_i,\ q< i\leq q+p $,
 $n =  \sum_{l=0}^{p} \#_{\a}(q+l)$ and  $z_1 = y$.
 By the Lemma \ref{lem1}, we have
  \begin{equation*}
 P_{\alpha}(\om) = \binom{\sum_{l=0}^{p} \#_{\a}(q+l) -1}{\#_{\a}(q)}_{\om} P_{\beta}(\om).
 \end{equation*}
We then carry out the same procedure for $\beta$. By induction, we obtain \eqref{Pf}.
\hfill $\square$

Using the expression \eqref{Pf}, we are able to prove the following corollary:
\begin{Cor}\label{cor1B}
Let $\a\in \cZ^k_{\geq}$. There exists a non zero rational function $R_{\a}(\om) \in \mathbb{Q}(\om)$ such that
\begin{equation}\label{pm0B}
P_{\a+m}( \om)=R_{\a}(\om)(1-\om^{\sum_{r=0}^{p-1}\#_{\a}(r-m)}) \quad \mbox{for all $m \in \mathbb{Z}$.}
\end{equation}
\end{Cor}
\begin{proof} It follows from \eqref{Pf} that
\begin{equation} \label{eqpB}
P_{\a}( \om) (1-\om^{\sum_{r=1}^{p}\#_{\a}(r)})=P_{\a-1}( \om)(1-\om^{\sum_{r=0}^{p-1}\#_{\a}(r)})
\end{equation}
Indeed, it is clear that
\begin{eqnarray*}
&&P_{\a}(\om) \prod_{r=1}^{p} \binom{\sum_{l=1}^{r} \#_{\a}(l)}{\#_{\a}(r)}_{\om}  \binom{ \sum_{l=0}^{p} \#_{\a}(l) -1}{\#_{\a}(0)}_{\om}\\
&&= P_{\a-1}(\om) \prod_{r=0}^{p-1} \binom{\sum_{l=0}^{r} \#_{\a}(l)}{\#_{\a}(r)}_{\om}  \binom{ \sum_{l=0}^{p} \#_{\a}(l) -1}{\#_{\a}(p)}_{\om}
\end{eqnarray*}
More explicitly (in the formula below out of simplicity we write an integer $n$ instead of $\om^n-1$),
\begin{eqnarray*}
&&P_{\a}(\om) \dfrac{ (\sum_{l=0}^{p} \#_{\a}(l) -1)!  (\sum_{l=1}^{p} \#_{\a}(l))! }{\#_{\a}(0)!...\#_{\a}(p)! (\sum_{l=1}^{p} \#_{\a}(l)-1)!}\\
&&= P_{\a-1}(\om) \dfrac{(\sum_{l=0}^{p-1} \#_{\a}(l))! (\sum_{l=0}^{p} \#_{\a}(l)-1)!}{\#_{\a}(0)!...\#_{\a}(p)! (\sum_{l=0}^{p-1} \#_{\a}(l)-1)!}
\end{eqnarray*}
and equation \eqref{eqpB} follows after simplifying these fractions.
We iterate \eqref{eqpB} and get for all $m,l \in \mathbb{Z}$,
\begin{equation*}
P_{\a+m}( \om) (1-\om^{\sum_{r=0}^{p-1}\#_{\a}(-l+r)})=P_{\a+l}( \om)(1-\om^{\sum_{r=0}^{p-1}\#_{\a}(-m+r)}).
\end{equation*}
We now choose $l$ such that $\sum_{r=0}^{p-1}\#_{\a}(-l+r)\neq 0$, and let
\begin{equation*}
R_{\a}(\om)=P_{\mu+l}( \om)(1-\om^{\sum_{r=0}^{p-1}\#_{\a}(-l+r)})^{-1}
\end{equation*}
and hence we obtain the required identity \eqref{pm0B}.
\end{proof}
We are now at the position to prove our main statement for the Narita-Itoh-Bogoyavlensky hierarchy.

{\bf Proof of Theorem \ref{BOGth}.}
We fix $k$ and let $u_{\tau}=Q^{(k)}$, which is a $k$-degree symmetry of the Narita-Itoh-Bogoyavlensky
equation given by \eqref{nsymsB}. Since $\mathcal{S}(\fI)=\fI$ we only need to show that
\begin{equation*}
\pi\left(\partial_{\tau}(uu_m-\om^{\delta_{m\leq p}} u_mu) \right)=0, \quad m \in \bbbn,
\end{equation*}
where $\delta_{m\leq p}=1$ if $1\leq m\leq p$ and $\delta_{m\leq p}=0$ if $m\geq p+1$.
This means that
\begin{equation*}
\pi\left(Q^{(k)} u_m+u Q^{(k)}_m-\om^{\delta_{m\leq p}} Q^{(k)}_m u-\om^{\delta_{m\leq p}}u_m Q^{(k)} \right)=0 .
\end{equation*}
We rewrite it in terms of $Y$. Here we simply drop its upper index of $Y^{(k)}$.
\begin{eqnarray}
&&\pi\left(uY_{m+1}u_m-\om^{\delta_{m\leq p}}Y_{m+1}u_mu-uu_mY_{m-p}+\om^{\delta_{m\leq p}}u_mY_{m-p}u \right.\nonumber\\
&&\qquad \left. +Y_1uu_m-\om^{\delta_{m\leq p}}u_mY_1u-uY_{-p}u_m+\om^{\delta_{m\leq p}}u_muY_{-p}\right) = 0. \label{idqb}
\end{eqnarray}
It is clear that, for any $\a\in \cZ^k_{\geq}$, we have
\begin{eqnarray*}
&&\pi(u u_{\a} u_m)= \om^{\sum_{k=1}^p (\#(k)-\#(-k))} \pi(u_{\a} u u_m), \\
&&\pi(u_m u_{\a} u )=\om^{-\delta_{m\leq p}} \om^{\sum_{k=1}^p (\#(m+k)-\#(m-k))} \pi(u_{\a} u u_m), \\
&&\pi(uu_m u_{\a}) =  \om^{\sum_{k=1}^p \left(\#(m+k)+ \#(k)-\#(-k)-\#(m-k)\right)} \pi(u_{\a} u u_m).
\end{eqnarray*}
Note that for all $l \in \mathbb{Z}$, we have
\begin{equation}
\pi(Y_l)=\pi(\cS^l Y) =\cS^l \pi(Y)
=\sum_{\a\in  \cZ^k_{\geq}} P_{\a}(\omega) u_{\a+l}
=\sum_{\a\in  \cZ^k_{\geq}} P_{\a-l}(\omega) u_{\a}.
\end{equation}
Hence, the left-handed side of \eqref{idqb} becomes
\begin{eqnarray*}
&&\sum_{\a\in  \cZ^k_{\geq}}\!\! \left(P_{\a-m-1}( \om)
-P_{\a-m+p}( \om) \om^{\sum_{k=1}^p (\#(m+k)-\#(m-k))}
\right) \left(\om^{\sum_{k=1}^p( \#(k)-\#(-k))}-1\right) \pi(u_{\a} u u_m)\\
&&+\sum_{\a\in  \cZ^k_{\geq}} \left(P_{\a-1}(\om) -P_{\a+p}(\om) \om^{\sum_{k=1}^p (\#(k)-\#(-k))}\right)
\left(1-\om^{\sum_{k=1}^p (\#(m+k)-\#(m-k))}\right)
\pi(u_{\a} u u_m)
\end{eqnarray*}
For any $\a\in \cZ^k_{\geq}$, we need to check that the coefficient of $\pi(u_{\a} u u_m)$ vanishes.
Using Corollary \ref{cor1B}, it becomes to compute
\begin{eqnarray*}
 &&\left(1-\om^{\sum_{k=1}^p \#(m+k)}-(1-\om^{\sum_{k=1}^p \#(m-k)})\om^{ \sum_{k=1}^p (\#(m+k)-\#(m-k))}
\right) \left(\om^{ \sum_{k=1}^p (\#(k)-\#(-k))}-1\right)\\
&&+\left(1-\om^{\sum_{k=1}^p \#(k)} -(1-\om^{\sum_{k=1}^p\#(-k)}) \om^{ \sum_{k=1}^p (\#(k)-\#(-k))}\right) \left(1-\om^{ \sum_{k=1}^p (\#(m+k)-\#(m-k))}\right),
\end{eqnarray*}
which equals zero after the simplification and thus we complete the proof.
\hfill $\square$

\end{document}